\documentclass[journal]{IEEEtran}

\ifCLASSINFOpdf

\else

\fi

\usepackage{mathrsfs}
\usepackage{amsmath}
\usepackage{amsfonts}
\usepackage{amsthm}
\usepackage{amssymb}
\usepackage{graphicx}
\usepackage{indentfirst}
\usepackage{array}
\usepackage{cite}
\usepackage{enumerate}
\usepackage[linesnumbered,ruled,vlined]{algorithm2e}
\usepackage{multirow}
\usepackage{multicol}
\usepackage{verbatim}
\usepackage{stfloats}
\newtheorem{theorem}{\bf{Theorem}}

\newtheorem{proposition}{\bf{Proposition}}

\newtheorem{definition}{\bf{Definition}}
\newtheorem{remark}{\bf{Remark}}

\begin{document}

\title{Polar-Coded Non-Orthogonal Multiple Access}

\author{Jincheng~Dai,~\IEEEmembership{Student Member,~IEEE},
        Kai~Niu,~\IEEEmembership{Member,~IEEE}, Zhongwei~Si,~\IEEEmembership{Member,~IEEE}, Chao~Dong,~\IEEEmembership{Member,~IEEE} and Jiaru Lin,~\IEEEmembership{Member,~IEEE}
\thanks{This work is supported by the National Natural Science Foundation of China (No. 61671080 \& No. 61171099), National High-tech R\&D Program (863 Program) (No. 2015AA01A709), BUPT-SICE Excellent Graduate Students Innovation Fund and Huawei HIRP project (No. HIRPO20140101). This paper was partially presented at the 2016 IEEE International Symposium on Information Theory (ISIT) \cite{Dai_polar_coded_NOMA}.}
\thanks{The authors are with the Key Laboratory of Universal Wireless Communications, Ministry of Education, Beijing University of Posts and Telecommunications, Beijing 100876, China (email: \{daijincheng, niukai, sizhongwei, dongchao, jrlin\}@bupt.edu.cn).}
\vspace{-1em}
}

\maketitle

\begin{abstract}
Non-orthogonal multiple access (NOMA) is one of the key techniques to address the high spectral efficiency and massive connectivity requirements for the fifth generation (5G) wireless system. To efficiently realize NOMA, we propose a joint design framework combining the polar coding and the NOMA transmission, which deeply mines the generalized polarization effect among the users. In this polar coded NOMA (PC-NOMA) framework, the original NOMA channel is decomposed into multiple bit polarized channels by using a three-stage channel transform, that is, user$\to$signal$\to$bit partitions. Specifically, for the first-stage channel transform, we design two schemes, namely sequential user partition (SUP) and parallel user partition (PUP). For the SUP, a joint successive cancellation detecting and decoding scheme is developed, and a search algorithm is proposed to schedule the NOMA detecting order which improves the system performance by enhanced polarization among the user synthesized channels. The ``worst-goes-first'' idea is employed in the scheduling strategy, and its theoretic performance is analyzed by using the polarization principle. For the PUP, a corresponding parallel detecting scheme is exploited to reduce the latency. The block error ratio performances over the additive white Gaussian noise channel and the Rayleigh fading channel indicate that the proposed PC-NOMA obviously outperforms the state-of-the-art turbo coded NOMA scheme due to the advantages of joint design between the polar coding and NOMA.
\end{abstract}


\IEEEpeerreviewmaketitle

\section{Introduction}

\subsection{Related Works}

\IEEEPARstart{M}{ultiple} access (MA) technique was regarded as the landmark of each generation of mobile communication systems from the first generation (1G) to the 4G. Recently, as a major candidate of multiple access technique in the 5G mobile communication system, non-orthogonal multiple access (NOMA) \cite{NOMA_original_comm_magazine} has been proposed to increase the system throughput and accommodate massive communication connections. Among the available NOMA techniques, sparse code multiple access (SCMA) \cite{SCMA_original_PIMRC} and pattern division multiple access (PDMA) \cite{PDMA_original_TVT} are the two typical code-domain NOMA schemes. In the transmitting end of SCMA, each user's bit stream is directly mapped into a sparse codeword, and multiple users share the same orthogonal physical resource elements (PREs), e.g. OFDM subcarriers. SCMA codebooks can flexibly adapt to various system requirements and support codeword overloading when the number of multiplexed users exceeds that of orthogonal resources. This is the key that SCMA can support massive connections and high spectral efficiency in the 5G system. Like the low-density parity-check (LDPC) codes, SCMA can be represented by a regular sparse factor graph with the variable nodes (VNs) representing users and the function nodes (FNs) representing PREs. For PDMA, its multiplexing manner in the code-domain is similar to that in SCMA; however, the number of edges connected to each user in the factor graph may be different \cite{PDMA_original_TVT}.

Among various types of SCMA and PDMA receivers, the successive interference cancellation (SIC) detector carries out detection and cancellation user by user. It performs well when the received signal-to-interference-and-noise-ratio (SINR) of each user is significantly different. However, it suffers from a performance loss when the SINR difference is not obvious among these users, in which case the error propagation may occur \cite{PDMA_original_TVT}. Moreover, when there exist correlations among the user access channels, the performance of SIC receiver will also degrade. Take the advantage of the sparsity, message passing algorithm (MPA) can be used in the factor graph and achieve a near-optimal performance \cite{SCMA_original_PIMRC,PDMA_original_TVT}. In addition, it is not sensitive to the SINR difference at the receiver and more robust to the channel correlations among users compared to the SIC receiver. Hence, MPA is considered as the primary multiuser detection scheme for SCMA and PDMA.

Turbo code is adopted in many existing NOMA systems for the error correction as the turbo coded NOMA (TC-NOMA) in the subsequent descriptions \cite{SCMA_original_PIMRC}. Furthermore, followed by the turbo principle in \cite{Wangxiaodong_turbo_iterative}, the turbo decoding and NOMA detecting can be jointly processed to develop a outer-loop iterative multiuser receiver, which significantly enhances the system performance \cite{PDMA_original_TVT,SCMA_outloop_iterations}. This essentially shows the potential advantages in the joint design of channel coding and NOMA transmission.

Polar codes proposed by Ar{\i}kan \cite{Arikan} are the first constructive codes (as opposed to random codes) that provably achieve the symmetric capacity of binary-input memoryless channels (BMCs). This capacity-achieving code family uses a technique called channel polarization. Shortly after the polar code was put forward, the channel polarization phenomenon has been rapidly found to be universal in many other signal processing problems, such as source coding, information secrecy and other scenarios. Particularly, in order to improve the spectral efficiency, Seidl \emph{et al.} have introduced a $2^m$-ary polar coded modulation scheme \cite{Polar_coded_modulation_seidl}. By considering the dependencies between the bits which are mapped to a single modulation symbol as a special kind of channel transform, the polar coded modulation (PCM) scheme is derived under the framework of two-stage channel transform. Note that with regard to the PCM, this two-stage channel transform concatenated manner accommodates a joint design of polar coding and modulation which allows one to describe the two techniques in a unified context rather than a simple combination.

\subsection{Motivation}

Note that the reliabilities of access users will demonstrate obvious differentia in the NOMA transmission. Especially for the PDMA, this phenomenon is elaborated as the ``disparate diversity order'' \cite{PDMA_original_TVT}. From the perspective of polarization, this reliability distinction among the users can be viewed as a generalized polarization effect. Therefore, the investigation of the generalized polarization in NOMA channel will be a key to improve the performance of NOMA system.

Guided by the generalized polarization idea, we design the framework of polar coded NOMA (PC-NOMA) system. The novelty of PC-NOMA is to allow for a joint optimization of the binary polar coding, the signal modulation and the NOMA transmission. Through a multi-stage channel transform concatenated manner, the polarization effect can be gradually enhanced. Finally, the NOMA channel will be elaborately split into a group of BMCs, whose capacities are trend to zero and one. Under this joint design framework, we analyze the polarization behavior of the NOMA channel and investigate the key features in the NOMA transceiver so as to achieve a significant performance improvement in NOMA system.

\subsection{Summary of Contributions}

The highlights of our contributions in this paper can be summarized as the following four aspects:

\noindent{\emph{1) Three-Stage Channel Transform Structure}}

Guided by the generalized polarization idea, we propose the three-stage channel transform concatenated manner, which constructs the framework of PC-NOMA. Specifically, the whole structure is described as user$\to$signal$\to$bit partitions and we focus on the first-stage design, i.e., the user partition, where the partition order and the partition structure will affect the polarization effect. In the second stage, the bit-interleaved code modulation (BICM) scheme is used to combat the channel fading. Then by performing the binary channel polarization transform, the bit synthesized channels are split into a series of bit polarized channels in the third stage. Essentially, the proposed three-stage channel transform structure facilitates the unified description of polar coding and NOMA transmission in a joint manner.

\noindent{\emph{2) Two Schemes in User Partition}}

We design two schemes, namely sequential user partition (SUP) and parallel user partition (PUP) for the first-stage channel transform. The SUP is designed from the perspective of information theory and follows the chain rule of mutual information, which aims for the highest possible spectrum efficiency. It sequentially decomposes the NOMA channel into a group of correlated user synthesized channels. For the PUP, the original NOMA channel will be parallel decomposed into a set of independent user synthesized channels, which aims to reduce the processing latency.

\noindent{\emph{3) JSC Receiver for The SUP Based PC-NOMA}}

For the SUP based PC-NOMA, we design a joint successive cancellation (JSC) detecting and decoding scheme. Since the partition order is related to the NOMA detecting order, we thus establish a fundamental theorem to obtain the optimal detecting order, named polar scheduling. On contrast to the traditional scheduling in MIMO or multiuser detections \cite{D_Tse_wireless_communications} which adopts the ``best-goes-first'' idea, the proposed polar scheduling strategy employs the ``worst-goes-first'' principle to promote the polarization effects among the user synthesized channels. Specifically, we provide the closed-form analytical solutions for the two-user case. For the general case with multiple users, we prove necessary optimality conditions for the polar scheduling. Guided by these necessary optimality conditions, we then derive a low complexity search algorithm for a near-optimal polar scheduling solution.

\noindent{\emph{4) PSC Receiver for The PUP Based PC-NOMA}}

Another concern in PC-NOMA design is the processing latency. We adjust the partition structure in the first stage and propose the PUP based PC-NOMA scheme. Correspondingly, we design a parallel and successive cancellation (PSC) hybrid detecting and decoding scheme in the receiver, whereby each user's detected messages are parallel sent to their associated polar decoders for further interference cancellation.

The remainder of the paper is organized as follows. Section II describes the notation conventions and the system model of NOMA transmission. Section III presents the framework of three-stage channel transform under the sequential and parallel user partitions, which can be seen as a joint design of the conventional binary channel polarization, the signal modulation and the NOMA transmission. The polar scheduling strategy in SUP is also discussed in this section. Then, guided by above two PC-NOMA design frameworks, the transceiver design of the proposed PC-NOMA schemes will be described in Section IV. Next, Section V evaluates the performance of the proposed PC-NOMA schemes under the additive white Gaussian noise (AWGN) channel and the Rayleigh fading channel. Finally, Section VI concludes this paper.

\section{Notations and System Model}

\subsection{Notation Conventions}

In this paper, we use calligraphic characters, such as ${\cal X}$, to denote sets. Let $\left| \cal X \right|$ denote the cardinality of the set $\cal X$. We write lowercase letters (e.g., $x$) to denote scalars. We use notation ${\bf{x}}$ to denote a vector and $x_i$ to denote the $i$-th element in ${\bf{x}}$. The set of binary, real and complex numbers are denoted by $\mathbb{B}$, $\mathbb{R}$ and $\mathbb{C}$, respectively. The bold letters, such as ${\bf{X}}$, denote matrices, and ${{\bf{X}}_{i,j}}$ is the element at the $i$-th row and the $j$-th column of matrix $\bf{X}$. For a diagonal matrix $diag\left( {\bf x} \right)$, its $i$-th diagonal element is $x_i$. The notations ${{\bf{X}}^T}$ and ${{\bf{x}}^T}$ stand for the transpose of matrix $\bf{X}$ and vector ${\bf{x}}$, respectively.

Specially, we use ${\cal N}({a},{b})$ to denote a Gaussian distribution with the mean $a$ and the variance $b$. Let ${\cal CN}({\bf a},{\bf B})$ stand for a complex Gaussian distribution, where $\bf a$ and $\bf B$ represent the mean vector and the covariance matrix, respectively. In addition, we will also use the uppercase letter (e.g., $Y$) to denote random variable and the lowercase letter $y$ to represent a realization. Furthermore, let the uppercase sans-serif letter (e.g., $\sf Y$) denote random vector and the lowercase bold letter $\bf y$ to represent a realization vector.

Throughout this paper, $\log \left(  \cdot  \right)$ means ``logarithm to base 2'', and $\ln \left(  \cdot  \right)$ stands for the ``natural logarithm to base $e$'', where the constant $e = 2.71828\cdots$.

\subsection{NOMA System Model}

Consider an uplink NOMA system including $V$ users, where each user has been assigned a different codebook. For each $v$-th user, the input block of the error-correcting encoder ${\bf u}_v$ includes $K_v$ bits. After the interleaving, the $N$-bits coded block is denoted by ${\bf c}_v$. In fading channels, the interleaver is indeed crucial to guarantee that consecutive coded bits affected by independent fades. Additionally, the $v$-th user's code rate is defined as ${R_v} = \frac{K_v}{N}$, and the overall code rate can be defined as $R = \frac{1}{V}\sum\nolimits_{v = 1}^V {{R_v}} $.

For the $v$-th user, by using the signal mapper, every $J$ coded bits in ${\bf c}_v$ form a vector
\begin{equation}\label{bitstream_map}
  \begin{aligned}
{{\bf{b}}_v} & = ({c_{v,(t - 1)J + 1}},{c_{v,(t - 1)J + 2}}, \cdots ,{c_{v,tJ}})\\
~ & = ({b_{v,1}},{b_{v,2}}, \cdots ,{b_{v,J}}),
\end{aligned}
\end{equation}
where $t = 1,2,\cdots,N/J$ denotes the time slot index, and $J$ is defined as the modulation order. Then, each vector ${\bf b}_v$ will be mapped to an $F$-dimensional codeword ${{\bf{x}}_v} = ({x_{v,1}},{x_{v,2}}, \cdots ,{x_{v,F}})$ by an NOMA mapper. We use ${{\bf{x}}_v^{(j)}}$ to denote the bit mapping from the $j$-th bit $b_{v,j}$ in ${\bf b}_v$. The number of nonzero elements in each codeword ${{\bf{x}}_v}$ is far less than $F$, namely, ${{\bf{x}}_v}$ is a sparse vector. The corresponding mapping function for the $v$-th user is defined as ${g_v}:{\mathbb{B}^J} \mapsto {{\cal X}_v}$, where ${{{\cal X}_v}} \in {\mathbb{C}^F}$ denotes the $v$-th user's codebook and $\left| {{{\cal X}_v}} \right| = M$. Hence, $J = {\log}M$. All the $V$ users' codebooks form the codebook set $\{ {{{\cal X}_v}} \}$. Then all the $V$ users' codewords are multiplexed over $F$ shared PREs, e.g. OFDM subcarriers. In the NOMA systems, we usually have $V > F$. Hence, overloading can be implemented by NOMA with a large number of users, which enables massive connectivity in 5G wireless systems. We define the system overloading factor (SOF) as $\eta  = V/F$ \cite{SCMA_original_PIMRC}.

The whole structure of NOMA can be represented by a binary $F \times V$ matrix $\bf{F}$, and the corresponding factor graph is denoted by ${\cal G}\left( {{\cal V},{\cal F}} \right)$, which contains $V$ variable nodes (VNs) associated with the users, and $F$ function nodes (FNs) associated with the PREs. We also use $F \times V$ to denote the NOMA system configuration. VN $v$ and FN $f$ are connected if and only if ${{\bf{F}}_{f,v}} = 1$. The set of VNs connected to FN $f$ is defined as ${{\cal V}_f} = \left\{ {v\left| {{{\bf{F}}_{f,v}} = 1} \right.} \right\}$ for $\forall f$, and ${d_f} = \left| {{{\cal V}_f}} \right|$ denotes the degree of FN $f$. Similarly, the set of FNs connected to VN $v$ is given as ${{\cal F}_v} = \left\{ {f\left| {{{\bf{F}}_{f,v}} = 1} \right.} \right\}$ for $\forall v$, and ${d_v} = \left| {{{\cal F}_v}} \right|$ denotes the degree of VN $v$.

\emph{Example 1:} Fig. \ref{factor_graph}(a) depicts the factor graph of a $V = 6$ user SCMA system, where the number of PREs $F = 4$ and SOF $\eta  = 150\%$. Its factor graph matrix is
\begin{equation}\label{6x4_factor_graph_matrix}
{\bf{F}} = \left[ {\begin{array}{*{20}{c}}
0&1&1&0&1&0\\
1&0&1&0&0&1\\
0&1&0&1&0&1\\
1&0&0&1&1&0
\end{array}} \right],
\end{equation}
where each VN's degree is ${d_v} = 2$ and each FN's degree is ${d_f} = 3$.

\emph{Example 2:} Fig. \ref{factor_graph}(b) depicts the factor graph of a $V = 3$ user PDMA system, where the number of PREs $F = 2$ and SOF $\eta  = 150\%$. Its factor graph matrix is
\begin{equation}\label{3x2_factor_graph_matrix}
{\bf{F}} = \left[ {\begin{array}{*{20}{c}}
1&1&0\\
1&0&1\\
\end{array}} \right],
\end{equation}
where each FN's degree is ${d_f} = 2$. However, each VN's degree is different, i.e., ${d_v}$ equals $1$ or $2$, which is the main difference between PDMA and SCMA.

\begin{figure}[htbp]
\setlength{\abovecaptionskip}{0.cm}
\setlength{\belowcaptionskip}{-0.cm}
  \centering{\includegraphics[scale=0.75]{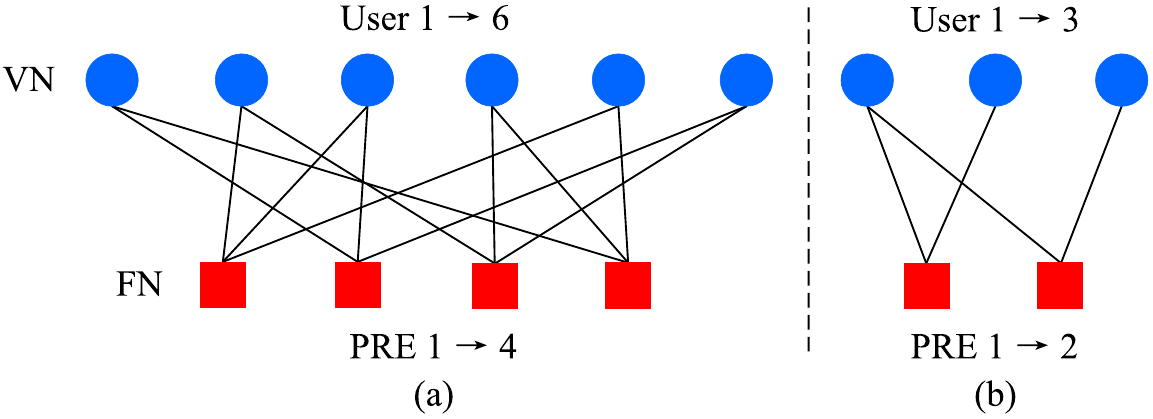}}
  \caption{The factor graph of NOMA systems. Figure (a) corresponds to a SCMA system with $V = 6$ users and $F = 4$ PREs. Figure (b) corresponds to a PDMA system with $V = 3$ users and $F = 2$ PREs.}\label{factor_graph}
\end{figure}

In the uplink code-domain multiplexing NOMA system, the received signal vector ${\bf{y}} = \left( {{y_1},{y_2}, \cdots ,{y_F}} \right)$ can be represented as
\begin{equation}\label{NOMA_rec_vec}
  {{\bf{y}}^T} = \sum\limits_{v = 1}^V {diag\left( {{{\bf{h}}_v}} \right){\bf{x}}_v^T}  + {{\bf{z}}^T},
\end{equation}
where ${{\bf{h}}_v} = ( {{h_{v,1}},{h_{v,2}}, \cdots ,{h_{v,F}}} )$ is the channel gain vector of the $v$-th user and ${\bf{z}}{{}\sim{}}{\cal C}{\cal N}\left( {{\bf 0},{N_0}{\bf{I}}} \right)$ is an AWGN vector with the mean vector $\bf 0$ and covariance matrix ${N_0}{\bf{I}}$, and $N_0$ denotes the variance of Gaussian noise. After receiving $\bf{y}$, multiuser detecting and channel decoding are employed to retrieve each user's information bits ${\bf{\hat{u}}}_v$. These generalized ``decoding'' operations can be accomplished in either a separately concatenated manner \cite{SCMA_original_PIMRC} or a jointly combined manner \cite{SCMA_outloop_iterations}.

\section{Channel Transform}

In this section, the concept of channel polarization transform is extended to NOMA transmission. The three-stage channel transform of PC-NOMA is proposed, which is distinguished by the sequential user partition (SUP) and the parallel user partition (PUP), respectively. We discuss the polar scheduling strategy under the SUP case, which follows the ``worst-goes-first'' idea. For the case of $V = 2$, we provide the explicitly optimal scheduling strategy. Then, when $V > 2$, the necessary conditions for the optimal scheduling are derived, which can efficiently reduce the optimal solution search complexity.

\subsection{Three-Stage Channel Transform under SUP}

For the code-domain multiplexing NOMA, we assume an ideal channel estimation in the receiver, i.e., the channel gain vector ${\bf{h}} = \left( {{{\bf{h}}_1},{{\bf{h}}_2}, \cdots ,{{\bf{h}}_V}} \right)$ is known by the receiver. The NOMA channel is denoted by $W:{{\cal X}} \mapsto {{\cal Y}}$, where $\cal X$ is set of transmitted vector ${\bf{x}} = \left( {{{\bf{x}}_1},{{\bf{x}}_2}, \cdots ,{{\bf{x}}_V}} \right)$ with $\left| {{{\cal X}}} \right| = {2^{JV}}$, and $\cal Y$ is the set of received vector $\bf{y}$. Given each user's codebook, every $JV$ bits vector $\bf{b}$ are mapped into a transmitted vector ${\bf{x}} \in {\cal X}$ under a specific one-to-one mapping called NOMA labeling
\begin{equation}\label{NOMA_labeling}
  L:{\bf{b}} = \left( {{{\bf{b}}_1},{{\bf{b}}_2}, \cdots ,{\bf{b}}{}_V} \right) \in {{\mathbb{B}}^{JV}} \mapsto {{\bf{x}} \in \cal X},
\end{equation}
where $L$ depends on each user's mapping function $g_v$. Hence, the NOMA channel $W$ is equivalently defined as $W:{ {\mathbb{B}}^{JV}} \mapsto {\cal Y}$, whose transition probability is
\begin{equation}\label{NOMA_transition_probabilities}
  W\left( {{\bf{y}}\left| {\bf b},{\bf h} \right.} \right) = W\left( {{\bf{y}}\left| {{L^{ - 1}}\left( {\bf{x}} \right)},{\bf h} \right.} \right),
\end{equation}
where $L^{-1}$ is the inverse mapping of $L$.

\begin{figure}[t]
\vspace{0.4em}
\setlength{\abovecaptionskip}{0.cm}
\setlength{\belowcaptionskip}{-0.cm}
  \centering{\includegraphics[scale=0.74]{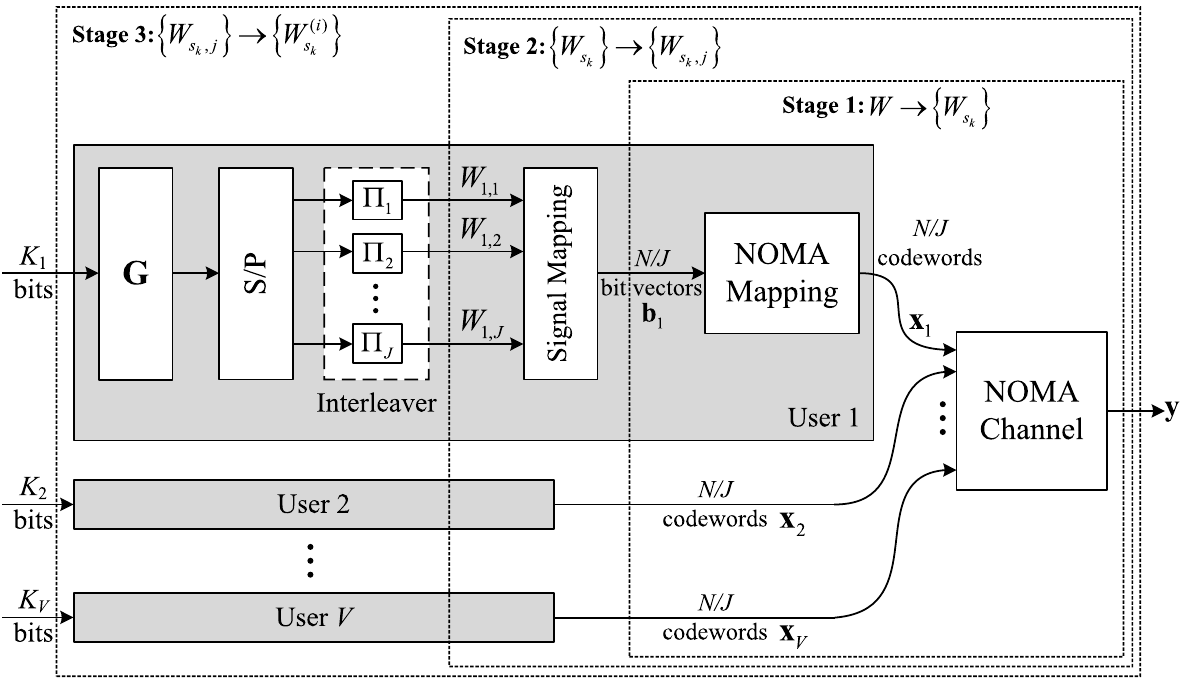}}
  \caption{Three-stage channel transform of PC-NOMA.}\label{SUP_three_stage_channel_transform_figure}
\vspace{-1em}
\end{figure}

For PC-NOMA, the NOMA channel $W$ is decomposed into a series of bit polarized channels under the three-stage channel transform. Let the user partition order be written as $\pi  = \{ {s_1},{s_2}, \cdots ,{s_V}\} $, and $\pi$ also corresponds to the detecting order in the receiver, where the ${s_1}$-th user is detected first, the ${s_2}$-th user is detected second, etc. The three-stage channel transform procedure can be written as
\begin{equation}\label{SUP_three_stage_channel_transform_procedure}
  W \to \left\{ {{W_{{s_k}}}} \right\} \to \left\{ {{W_{{s_k},j}}} \right\} \to \left\{ {W_{{s_k}}^{\left( i \right)}} \right\},
\end{equation}
which is depicted in Fig. \ref{SUP_three_stage_channel_transform_figure}. Based on the SUP scheme, the first stage channel transform is closely related to that of multilevel coding (MLC) in the PCM on a conceptual level. Then using the BICM approach, with the help of interleavers ${\Pi _j}$, $j = 1,2,\cdots,J$, the signal partition is performed in the second stage. Finally, the bit polarization transform is employed in the third stage. On the whole, this three-stage channel transform constructs a MLC/BICM hybrid structure.

\noindent{\emph{Stage 1: Sequential user partition}}

Given the user partition order $\pi  = \{ {s_1},{s_2}, \cdots ,{s_V}\} $, with regard to the chain rule of mutual information \cite{Cover_elemnts_of_information_theory}, the average mutual information (AMI) between $\bf b$ and $\bf y$ is written as
\begin{equation}\label{NOMA_channel_chain_rule}
  \begin{aligned}
    & I\left( {{{\sf B}_1},{{\sf B}_2}, \cdots ,{{\sf B}_V};{\sf Y}\left| {\sf H} \right.} \right) =  \\
    & \sum\limits_{k = 1}^V {{I_{{s_{k}}}}} = \sum\limits_{k = 1}^V {I\left( {{{\sf B}_{{s_{k}}}};{\sf Y}\left| {{{\sf B}_{{s_1}}},{{\sf B}_{{s_2}}}, \cdots ,{{\sf B}_{{s_{k - 1}}}},{\sf H}} \right.} \right)},
    \end{aligned}
\end{equation}
where $\sf Y$ and $\sf H$ represent the random vectors corresponding to $\bf y$ and $\bf h$, respectively. Additionally, ${\sf B}_{s_{k}}$ denotes the random vector associated with ${\bf b}_{s_{k}}$.

Note that each user's AMI ${I_{{s_{k}}}}$ can be formulated as
\begin{equation}\label{NOMA_mutual_information_transform_SUP}
  \begin{aligned}
    {I_{{s_{k}}}} = & I\left( {{{\sf B}_{{s_{k}}}};{\sf Y}\left| {{{\sf B}_{{s_1}}},{{\sf B}_{{s_2}}}, \cdots ,{{\sf B}_{{s_{k - 1}}}},{\sf H}} \right.} \right)\\
     = & H\left( {{{\sf B}_{{s_{k}}}}\left| {{{\sf B}_{{s_1}}},{{\sf B}_{{s_2}}}, \cdots ,{{\sf B}_{{s_{k - 1}}}},{\sf H}} \right.} \right) - \\
     & H\left( {{{\sf B}_{{s_{k}}}}\left| {{{\sf B}_{{s_1}}},{{\sf B}_{{s_2}}}, \cdots ,{{\sf B}_{{s_{k - 1}}}},{\sf Y},{\sf H}} \right.} \right) \\
     = & H\left( {{{\sf B}_{{s_{k}}}}\left| {\sf H} \right.} \right) - H\left( {{{\sf B}_{{s_{k}}}}\left| {{{\sf B}_{{s_1}}},{{\sf B}_{{s_2}}}, \cdots ,{{\sf B}_{{s_{k - 1}}}},{\sf Y},{\sf H}} \right.} \right)\\
     = & I\left( {{{\sf B}_{{s_{k}}}};{{\sf B}_{{s_1}}},{{\sf B}_{{s_2}}}, \cdots ,{{\sf B}_{{s_{k - 1}}}},{\sf Y}\left| {\sf H} \right.} \right),
    \end{aligned}
\end{equation}
where $H\left(  \cdot  \right)$ stands for the entropy function. We attribute the chain of equalities in Eq. (\ref{NOMA_mutual_information_transform_SUP}) to the fact that each user's transmitting bits are independent with each other. So given the user partition order $\pi$, the NOMA transmission process can be seen as a kind of channel transform. Eq. (\ref{NOMA_mutual_information_transform_SUP}) indeed indicates that each user's data stream is detected in a SIC manner at the receiver, then total $V$ correlated channels ${{W}_{s_k}} : {{\mathbb{B}}^J} \mapsto {\cal Y} \times {{\mathbb{B}}^{J\left( {k - 1} \right)}}$ are obtained, whose transition probability function is written as
\begin{equation}\label{transition_prob_SUP_user_synthesized_channel}
  \begin{aligned}
    & {W_{{s_k}}}\left( {{\bf{y}},{{\bf{b}}_{{s_1}}},{{\bf{b}}_{{s_2}}}, \cdots ,{{\bf{b}}_{{s_{k - 1}}}} \left| {{{\bf{b}}_{{s_k}}},{\bf{h}}} \right.} \right) \\
    & = \sum\limits_{\left( {{{\bf{b}}_{{s_{k + 1}}}}, \cdots ,{{\bf{b}}_{{s_V}}}} \right) \in {{\mathbb B}^{J\left( {V - k} \right)}}} {\left( {\frac{1}{{{2^{J\left( {V - 1} \right)}}}} \cdot W\left( {{\bf{y}}\left| {\bf{b}},{\bf h} \right.} \right)} \right)}.
    \end{aligned}
\end{equation}
Assuming the input binary vectors are equiprobable, Eq. (\ref{NOMA_mutual_information_transform_SUP}) is therefore referred as the symmetric capacity of ${W_{{s_k}}}$. Note that for any detecting order $\pi$, the AMI of NOMA channel $W$ is preserved, i.e.,
\begin{equation}\label{NOMA_channel_capacity_sum}
  \begin{aligned}
    I\left( W \right) = \sum\limits_{k = 1}^V {I\left( {{W_{{s_k}}}} \right)} & = \sum\limits_{k = 1}^V {{I_{{s_k}}}} \\
    ~ & = I\left( {{{\sf B}_1},{{\sf B}_2}, \cdots ,{{\sf B}_V};{\sf Y}\left| {\sf H} \right.} \right),
    \end{aligned}
\end{equation}
which follows the property of AMI chain rule \cite{Cover_elemnts_of_information_theory}, and the function $I\left(  \cdot  \right)$ denotes the symmetric capacity.

\begin{definition}
  \emph{An order-$V$ \emph{sequential user partition} ($V$-SUP) of NOMA channel $W$ is defined as
    \begin{equation}\label{SUP_channel_transform}
      W \to \left\{ {{W_{{s_1}}},{W_{{s_2}}}, \cdots ,{W_{{s_V}}}} \right\},
    \end{equation}
  which maps $W$ to an ordered set of $V$ binary vector input channels ${W_{{s_k}}}$ with $k=1,2,\cdots,V$, which are referred as the \emph{user synthesized channels}.}
\end{definition}

For a given NOMA channel $W$, the $V$-SUP is characterized by the user partition order $\pi$, and the number of all possible partition orders equals $\left( V! \right)$.

\noindent{\emph{Stage 2: Signal partition}}

In the second-stage transform, each user synthesized channel ${{W}_{s_k}}$ will be further transformed into a set of binary input channels $\{ {W_{{s_k},j}}\} $ with $j = 1,2,\cdots,J$, which are referred as the \emph{bit synthesized channels}.

Similar to channel transform scheme in \cite{Polar_coded_modulation_seidl}, we use BICM scheme in the signal partition. In fading channels, this BICM is indeed crucial to guarantee that consecutive coded bits affected by independent fades. Note that
\begin{equation}\label{SUP_signal_partition_capacity_equation}
  \begin{aligned}
    & I\left( {{W_{{s_k}}}} \right) = {I_{{s_k}}}\\
    & = \sum\limits_{j = 1}^J {I\left( {{B_{{s_k},j}};{\sf Y}\left| {{{\sf B}_{{s_1}}}, \cdots ,{{\sf B}_{{s_{k - 1}}}},{B_{{s_k},1}}, \cdots ,{B_{{s_k},j - 1}},{\sf H}} \right.} \right)},
    \end{aligned}
\end{equation}
where ${{B_{{s_k},j}}}$ denotes the random variable corresponding to ${{b_{{s_k},j}}}$ which is the $j$-th bit of ${\bf b}_{s_k}$ with $j = 1,2,\cdots,J$. By assuming the input bits of $W_{s_k}$ are uniformly distributed, each $2^J$-ary input vector channel under the interleaver is equivalent to a group of $J$ parallel independent BMCs with transition probabilities
\begin{equation}\label{transition_prob_SUP_bit_synthesized_channel}
  \begin{aligned}
     & {W_{{s_k},j}}\left( {{\bf{y}},{{\bf{b}}_{{s_1}}},{{\bf{b}}_{{s_2}}}, \cdots ,{{\bf{b}}_{{s_{k - 1}}}}\left| {{b_{{s_k},j}} = b,{\bf{h}}} \right.} \right) = \\
     & \sum\limits_{\begin{smallmatrix}
     {{\bf{b}}_{{s_k}}} \in {{\mathbb B}^J},\\
     {b_{{s_k},j}} = b
     \end{smallmatrix}} {\left( {\frac{1}{{{2^{J - 1}}}} \cdot {W_{{s_k}}}\left( {{\bf{y}},{{\bf{b}}_{{s_1}}},{{\bf{b}}_{{s_2}}}, \cdots ,{{\bf{b}}_{{s_{k - 1}}}}\left| {{{\bf{b}}_{{s_k}}},{\bf{h}}} \right.} \right)} \right)},
    \end{aligned}
\end{equation}
where $b \in {\mathbb B}$. According to Eq. (\ref{transition_prob_SUP_bit_synthesized_channel}), we have
\begin{equation}\label{SUP_signal_partition_capacity_inequation}
  \begin{aligned}
    & \sum\limits_{j = 1}^J {I\left( {{W_{{s_k},j}}} \right)} = \sum\limits_{j = 1}^J {I\left( {{B_{{s_k},j}};{\sf Y}\left| {{{\sf B}_{{s_1}}}, \cdots ,{{\sf B}_{{s_{k - 1}}}},{\sf H}} \right.} \right)}\\
    & \le \sum\limits_{j = 1}^J {I\left( {{B_{{s_k},j}};{\sf Y}\left| {{{\sf B}_{{s_1}}}, \cdots ,{{\sf B}_{{s_{k - 1}}}},{B_{{s_k},1}}, \cdots ,{B_{{s_k},j - 1}},{\sf H}} \right.} \right)}. \\
    \end{aligned}
\end{equation}

Essentially, Eq. (\ref{transition_prob_SUP_user_synthesized_channel}) and Eq. (\ref{transition_prob_SUP_bit_synthesized_channel}) transform the original NOMA channel $W$ into a set of $JV$ binary input channels $\{ W _{{s_v},j} \}$. Then, combining Eq. (\ref{SUP_signal_partition_capacity_equation}) and Eq. (\ref{SUP_signal_partition_capacity_inequation}), from the perspective of information theory, we have
\begin{equation}\label{SUP_second_stage_capacity_inequality}
  I\left( W \right) = \sum\limits_{k = 1}^V {I\left( {{W_{{s_k}}}} \right)}  \ge \sum\limits_{k = 1}^V {\sum\limits_{j = 1}^J {I\left( {{W_{{s_k},j}}} \right)} },
\end{equation}
which indicates that the BICM structure in the second stage channel transform brings some capacity loss when it enhances the robustness of PC-NOMA.

\noindent{\emph{Stage 3: Bit partition}}

In the third stage, the channel transform is named as the \emph{bit partition}, and the resulting BMCs $\{ W_{s_k}^{\left( i \right)}\} $ are referred as the \emph{bit polarized channels}. Similar to the bit polarization in the second stage described in \cite{Polar_coded_modulation_seidl}, for each user, by performing the binary-input channel polarization $\bf{G}$ on $N/J$ uses of the resulting $J$ parallel BMCs $\{ W _{{s_k},j}\}$, a total of $NV$ polarized BMCs $\{ W_{s_k}^{\left( i \right)}\} $ will be obtained, where $k = 1,2, \cdots ,V$ and $i = 1,2, \cdots ,N$.

\subsection{Polar Scheduling Strategy}

For the PC-NOMA, the second and third stage channel transforms are of fixed structures. However, for the SUP in the first stage, different partition orders $\pi$ will lead to various distributions of $\{ W_{s_k}\} $, which further affect the polarization effect of $\{ W_{s_k}^{\left( i \right)}\} $. To describe the distribution of $\{ W_{s_k}\} $, we focus on two properties of $\{ W_{s_k}\} $, i.e., the mean and the variance of the user synthesized channel capacities. Given the NOMA channel $W$ and the partition order $\pi$, they are defined respectively as
\begin{equation}\label{SUP_mean_definition}
  {M_\pi }\left( W \right) = \frac{1}{V}\sum\limits_{k = 1}^V {I\left( {{W_{{s_k}}}} \right)}  = \frac{1}{V}I\left( W \right),
\end{equation}
\begin{equation}\label{SUP_variance_definition}
  {V_\pi }\left( W \right) = \frac{1}{V}\sum\limits_{k = 1}^V {I^2{{\left( {{W_{{s_k}}}} \right)}}}  - {M_\pi^2}{\left( W \right)},
\end{equation}
where the functions ${M_\pi }\left( \cdot \right)$ and ${V_\pi }\left( \cdot \right)$ denote the mean and the variance, respectively. Eq. (\ref{SUP_mean_definition}) indicates that the mean ${M_\pi }\left( W \right)$ just relies on the original NOMA channel $W$ and it is partition order independent, which makes the variance ${V_\pi }\left( W \right)$ become an effective criterion for evaluating the polarization effect among these user synthesized channels $\{ W_{s_k}\} $.

Given the NOMA channel $W$, the optimal partition order maximizes the instantaneous variance ${V_\pi }\left( \cdot \right)$, that is,
\begin{equation}\label{polar_sheduling_optimize_target}
  {\pi ^*} = \mathop {\arg \max }\limits_\pi  {V_\pi }\left( W \right).
\end{equation}
Unfortunately, the exact analysis of ${I\left( {{W_{{s_k}}}} \right)}$ is actually a computationally intensive task. However, recall that the partition order is equivalent to the SIC detecting order in the receiver. In addition, we note that ${I\left( {{W_{{s_k}}}} \right)}$ is proportional to its detecting signal-to-interference-and-noise-ratio (SINR). Therefore, each user's detecting SINR under the SIC manner will be employed as an alternate to derive the optimal partition order.

Given detecting order $\pi  = \{ {s_1},{s_2}, \cdots ,{s_V}\} $, the indices set corresponding to the users having been detected and decoded can be defined as
\begin{equation}\label{detected_set_definition}
  {{\cal D}_k} = \left\{ {{s_1},{s_2}, \cdots ,{s_k}} \right\},
\end{equation}
where $k$ stands for detecting level in the SIC receiver with $k = 1,2,\cdots,V$, and we have ${{\cal D}_0} = \varnothing $. We use ${{\cal D}_k^c}$ to denote the complementary set of ${{\cal D}_k}$ with respect to the universal set ${\cal V} = \left\{ {1,2, \cdots ,V} \right\}$. When ${{\cal D}_k}$ has been determined, the \emph{pruned factor graph} is obtained by removing all the VNs $v \in {{\cal D}_k}$ and their associated edges from the original factor graph ${\cal G}\left( {{\cal V},{\cal F}} \right)$. Then, the corresponding factor graph matrix is denoted by ${{\bf{F}}^{\left( {{{\cal D}_k}} \right)}}$, which is generated by setting these columns with indices $v \in {{\cal D}_k}$ as all-zero columns. For instance, when ${{\cal D}_2} = \left\{ {1,2} \right\}$, the factor graph matrix in Eq. (\ref{6x4_factor_graph_matrix}) is modified as
\begin{equation}
  {{\bf{F}}^{\left( {{{\cal D}_2}} \right)}} = \left[ {\begin{array}{*{20}{c}}
    0&0&1&0&1&0\\
    0&0&1&0&0&1\\
    0&0&0&1&0&1\\
    0&0&0&1&1&0
    \end{array}} \right].
\end{equation}
Given ${{\bf{F}}^{\left( {{{\cal D}_k}} \right)}}$, its degrees of FN and VN are denoted by $d_f^{\left( {{{\cal D}_k}} \right)}$ and $d_v^{\left( {{{\cal D}_k}} \right)}$, respectively. Apparently, for any $v \in {{\cal D}_k}$, we have $d_v^{\left( {{{\cal D}_k}} \right)} = 0$. For any $1 \le k \le V$, the degree of FN will be less than the original one, that is, $d_f^{\left( {{{\cal D}_k}} \right)} \le {d_f}$.

Given each user's transmitting power $P$, their instantaneous SINR in each detecting level is defined as follows.

\begin{definition}\label{SINR_definition}
\emph{
  After ${\cal D}_k$ has been determined, the undetected user's \emph{instantaneous SINR} is written as
  \begin{equation}\label{SINR_definition_formula}
  \begin{aligned}
    & \gamma _v^{\left( {{{\cal D}_k}} \right)} = \\
    & \sum\limits_{f \in {{\cal F}_v}} {\frac{P}{{\left( {d_f^{\left( {{{\cal D}_k}} \right)} - 1} \right)P + {N_0}}}} = \sum\limits_{f \in {{\cal F}_v}} {\frac{1}{{\left( {d_f^{\left( {{{\cal D}_k}} \right)} - 1} \right) + \lambda }}},
  \end{aligned}
  \end{equation}
  where $v \in {{\cal D}_k^c}$ and the constant $\lambda  = \frac{{{N_0}}}{P}$.
  }
\end{definition}

When ${{\cal D}_{k-1}}$ has been determined, we need to choose one user from indices $v \in {{\cal D}_{k-1}^c}$ to detect, which is denoted by $s_k$. Its \emph{detecting SINR} is written as
\begin{equation}
  {\gamma _{{s_k}}} = \gamma _{{s_k}}^{\left( {{{\cal D}_{k - 1}}} \right)}.
\end{equation}
After completing above SIC detecting procedure, the detecting SINR sequence ${\cal T} = \{ {\gamma _{{s_1}}},{\gamma _{{s_2}}}, \cdots ,{\gamma _{{s_V}}}\} $ is obtained. We are also interested in the mean and the variance over $\cal T$, which are defined respectively as
\begin{equation}\label{SINR_mean_definition}
  {M_\pi }\left( \cal T \right) = \frac{1}{V}\sum\limits_{k = 1}^V {{\gamma _{{s_k}}}}  = \frac{1}{V}\sum\limits_{f = 1}^F {\sum\limits_{t = 1}^{{d_f}} {\frac{1}{{\left( {{d_f} - t} \right) + \lambda }}} },
\end{equation}
\begin{equation}\label{SINR_variance_definition}
  {V_\pi }\left( \cal T \right) = \frac{1}{V}\sum\limits_{k = 1}^V {\gamma _{{s_k}}^2}  - {M_\pi^2 }\left( \cal T \right) \propto \sum\limits_{k = 1}^V {\gamma _{{s_k}}^2} = {\Theta _\pi }.
\end{equation}
Clearly, according to Eq. (\ref{SINR_mean_definition}), the mean of $\cal T$ depends only on the original factor graph structure, rather than the particular detecting order $\pi$. This property is consistent with the analysis in Eq. (\ref{SUP_mean_definition}), which verifies that the detecting SINR can act as a rational alternate with respect to the capacity analysis. The optimal detecting order maximizes the instantaneous variance ${V_\pi }\left( \cal T \right)$, that is,
\begin{equation}\label{SINR_optimal_target}
  {\pi ^*} = \mathop {\arg \max }\limits_\pi  {V_\pi }\left( \cal T \right) = \mathop {\arg \max }\limits_\pi {\Theta _\pi }.
\end{equation}

When the number of users is $V = 2$, the optimal detecting order is as follows.

\begin{proposition}\label{Proposition_SINR_polar_scheduling_two_users}
\emph{
  The optimal partition (or detecting) order in the two-user PC-NOMA system is ``worst-goes-first'', which is expressed as
  \begin{equation}\label{SINR_polar_scheduling_two_users_formula}
    {\pi ^*} = \mathop {\arg \max }\limits_\pi  {\Theta _\pi }  = \left\{ {1,2} \right\}{\rm{~\text{iff}~}}\gamma _{{1}}^{\left( {{{\cal D}_0}} \right)} \le \gamma _{{2}}^{\left( {{{\cal D}_0}} \right)}.
  \end{equation}
  }
\end{proposition}
\begin{proof}
  Given two detecting orders ${\pi _1} = \left\{ {1,2} \right\}$ and ${\pi _2} = \left\{ {2,1} \right\}$, we first prove the necessity of the condition, which corresponds to the ``only if'' part in ``iff''. If ${\pi _1}$ is adopted, we have
  \begin{equation}\label{proposition_proof_1}
    {\gamma _1} = \gamma _1^{\left( {{{\cal D}_0}} \right)} = \sum\limits_{f \in {{\cal F}_1}} {\frac{1}{{\left( {d_f^{\left( {{{\cal D}_0}} \right)} - 1} \right) + \lambda }}},
  \end{equation}
  where ${d_f^{\left( {{D_0}} \right)}}$ indeed equals the original $d_f$ due to ${{\cal D}_0} = \varnothing $. In addition, note that
  \begin{equation}\label{proposition_proof_2}
    \begin{aligned}
    & {\gamma _2} = \gamma _2^{\left( {{{\cal D}_1}} \right)} = \sum\limits_{f \in {{\cal F}_2}} {\frac{1}{{\left( {d_f^{\left( {{{\cal D}_{1}}} \right)} - 1} \right) + \lambda }}} \\
    & = \sum\limits_{f \in {{\cal F}_2} \cap {{\cal F}_1}} {\frac{1}{{\left( {d_f^{\left( {{{\cal D}_0}} \right)} - 2} \right) + \lambda }}}  + \sum\limits_{f \in {{\cal F}_2} - {{\cal F}_1}} {\frac{1}{{\left( {d_f^{\left( {{{\cal D}_0}} \right)} - 1} \right) + \lambda }}} \\
    & = \sum\limits_{f \in {{\cal F}_2}} {\frac{1}{{\left( {d_f^{\left( {{{\cal D}_0}} \right)} - 1} \right) + \lambda }}}  + \Delta  = \gamma _2^{\left( {{{\cal D}_0}} \right)} + \Delta,
    \end{aligned}
  \end{equation}
  where $\Delta$ is written as
  \begin{equation}\label{proposition_proof_3}
    \Delta  = \sum\limits_{f \in {{\cal F}_2} \cap {{\cal F}_1}} {\left( {\frac{1}{{\left( {d_f^{\left( {{{\cal D}_0}} \right)} - 2} \right) + \lambda }} - \frac{1}{{\left( {d_f^{\left( {{{\cal D}_0}} \right)} - 1} \right) + \lambda }}} \right)} > 0.
  \end{equation}
  Therefore, we can conclude that
  \begin{equation}\label{proposition_proof_SINR_sum_pi1}
    {\Theta _{{\pi _1}}} = {\left( {\gamma _1} \right)^2} + {\left( {\gamma _2} \right)^2} = {\left( {\gamma _1^{\left( {{{\cal D}_0}} \right)}} \right)^2} + {\left( {\gamma _2^{\left( {{{\cal D}_0}} \right)} + \Delta } \right)^2}.
  \end{equation}
  Next, the detecting order is configured as ${\pi}_2$. In the first detecting level, we have
  \begin{equation}\label{proposition_proof_4}
    {\gamma _2} = \gamma _2^{\left( {{{\cal D}_0}} \right)} = \sum\limits_{f \in {{\cal F}_2}} {\frac{1}{{\left( {d_f^{\left( {{{\cal D}_0}} \right)} - 1} \right) + \lambda }}}.
  \end{equation}
  In the second detecting level, it can be derived that
  \begin{equation}\label{proposition_proof_5}
    \begin{aligned}
    & {\gamma _1} = \gamma _1^{\left( {{{\cal D}_1}} \right)} = \sum\limits_{f \in {{\cal F}_1}} {\frac{1}{{\left( {d_f^{\left( {{{\cal D}_{1}}} \right)} - 1} \right) + \lambda }}} \\
    & = \sum\limits_{f \in {{\cal F}_1} \cap {{\cal F}_2}} {\frac{1}{{\left( {d_f^{\left( {{{\cal D}_0}} \right)} - 2} \right) + \lambda }}}  + \sum\limits_{f \in {{\cal F}_1} - {{\cal F}_2}} {\frac{1}{{\left( {d_f^{\left( {{{\cal D}_0}} \right)} - 1} \right) + \lambda }}} \\
    & = \sum\limits_{f \in {{\cal F}_1}} {\frac{1}{{\left( {d_f^{\left( {{{\cal D}_0}} \right)} - 1} \right) + \lambda }}}  + \Delta  = \gamma _1^{\left( {{{\cal D}_0}} \right)} + \Delta,
    \end{aligned}
  \end{equation}
  where $\Delta$ is the same as that in Eq. (\ref{proposition_proof_3}). Additionally, we can also conclude that
  \begin{equation}\label{proposition_proof_SINR_sum_pi2}
    {\Theta _{{\pi _2}}} = {\left( {\gamma _2} \right)^2} + {\left( {\gamma _1} \right)^2} = {\left( {\gamma _2^{\left( {{{\cal D}_0}} \right)}} \right)^2} + {\left( {\gamma _1^{\left( {{{\cal D}_0}} \right)} + \Delta } \right)^2}.
  \end{equation}

  Since the optimal detecting order is assumed as ${\pi ^*} = \left\{ {1,2} \right\}$, we can derive the following inequalities
  \begin{equation}\label{proposition_inequalities}
    \begin{aligned}
     & {\Theta _{{\pi _1}}} \ge {\Theta _{{\pi _2}}}\\
     \Rightarrow & {\left( {\gamma _1^{\left( {{{\cal D}_0}} \right)}} \right)^2} + {\left( {\gamma _2^{\left( {{{\cal D}_0}} \right)} + \Delta } \right)^2} \ge {\left( {\gamma _2^{\left( {{{\cal D}_0}} \right)}} \right)^2} + {\left( {\gamma _1^{\left( {{{\cal D}_0}} \right)} + \Delta } \right)^2}\\
     \Rightarrow & \gamma _2^{\left( {{{\cal D}_0}} \right)} \ge \gamma _1^{\left( {{{\cal D}_0}} \right)},
    \end{aligned}
  \end{equation}
  which follows Eq. (\ref{proposition_proof_SINR_sum_pi1}) and Eq. (\ref{proposition_proof_SINR_sum_pi2}). The chain of inequalities in Eq. (\ref{proposition_inequalities}) indeed proves the the necessity of the condition in \emph{Proposition \ref{Proposition_SINR_polar_scheduling_two_users}}. Apparently, we observe that the identical chain of inequalities holds on in the opposite direction. Therefore, the sufficiency of the condition in \emph{Proposition \ref{Proposition_SINR_polar_scheduling_two_users}} can also be trivially proved.
\end{proof}

\emph{Proposition \ref{Proposition_SINR_polar_scheduling_two_users}} indicates a \emph{polar scheduling} strategy. It makes the reliability difference between ``bad'' user and ``good'' user become more significant, which essentially brings in enhanced polarization effect among the user synthesized channels. This scheduling strategy obeys the \emph{polarization principle} in the design of PC-NOMA, i.e., maximizing the variance of user synthesized channel capacities by the choice of detecting order $\pi$. Furthermore, from the proof process, it can be concluded that this scheduling strategy is signal-to-noise-ratio (SNR) independent, since it only relies on the NOMA superposition manner, i.e., the factor graph structure.

In general cases with $V > 2$, the following necessary optimality conditions can be formulated.

\begin{theorem}\label{Theroem_SINR_polar_scheduling_multiple_users}
\emph{
  For the case with $V > 2$, the optimal detecting order ${\pi ^*} = \left\{ {{s_1},{s_2}, \cdots ,{s_V}} \right\}$ of the PC-NOMA system must satisfy the following necessary optimality conditions:
  \begin{equation}\label{SINR_polar_scheduling_multiple_users_formula}
    \gamma _{{s_k}}^{\left( {{{\cal D}_{k - 1}}} \right)} \le \gamma _{{s_{k + 1}}}^{\left( {{{\cal D}_{k - 1}}} \right)}, \forall 1 \le k \le \left( {V - 1} \right),
  \end{equation}
  where the two SINRs are given in \emph{Definition \ref{SINR_definition}}.
  }
\end{theorem}
\begin{proof}
  Considering two orders
  \begin{equation}
    \left\{ \begin{aligned}
    & {\pi _1} = {\pi ^*} = \left\{ {{s_1},{s_2}, \cdots ,{s_k},{s_{k + 1}}, \cdots ,{s_V}} \right\},\\
    & {\pi _2} = \left\{ {{s_1},{s_2}, \cdots ,{s_{k + 1}},{s_k}, \cdots ,{s_V}} \right\},
    \end{aligned} \right.
  \end{equation}
  where $s_k$ and $s_{k+1}$ swap their positions. Then the difference between ${\Theta _{{\pi _1}}}$ and ${\Theta _{{\pi _2}}}$ is solely determined by the contribution of $s_k$-th user and $s_{k+1}$-th user, since this position swapping does not affect the contributions of other users. That means
  \begin{equation}
    {\Theta _{{\pi _1}}} - {\gamma _{{s_k},{\pi _1}}} - {\gamma _{{s_{k + 1}},{\pi _1}}} = {\Theta _{{\pi _2}}} - {\gamma _{{s_k},{\pi _2}}} - {\gamma _{{s_{k + 1}},{\pi _2}}},
  \end{equation}
  where ${\gamma _{{s_k},{\pi _1}}}$ and ${\gamma _{{s_k},{\pi _2}}}$ denote the $\gamma _{{s_k}}$ obtained under the detecting orders ${\pi}_1$ and ${\pi}_2$, respectively, and similar definitions for ${\gamma _{{s_{k+1}},{\pi _1}}}$ and ${\gamma _{{s_{k+1}},{\pi _2}}}$. Apply \emph{Proposition \ref{Proposition_SINR_polar_scheduling_two_users}} to the ${{s_k}}$-th user and the ${{s_{k+1}}}$-th user, we can also derive the following inequalities:
  \begin{equation}\label{theorem_inequalities}
    \begin{aligned}
    & {\Theta _{{\pi _1}}} \ge {\Theta _{{\pi _2}}}\\
    \Rightarrow & {\gamma _{{s_k},{\pi _1}}} + {\gamma _{{s_{k + 1}},{\pi _1}}} \ge {\gamma _{{s_k},{\pi _2}}} + {\gamma _{{s_{k + 1}},{\pi _2}}}\\
     \Rightarrow & {\left( {\gamma _{{s_k}}^{\left( {{{\cal D}_{k - 1}}} \right)}} \right)^2} + {\left( {\gamma _{{s_{k + 1}}}^{\left( {{{\cal D}_{k - 1}}} \right)} + \Delta } \right)^2} \ge \\
     & {\left( {\gamma _{{s_{k + 1}}}^{\left( {{{\cal D}_{k - 1}}} \right)}} \right)^2} + {\left( {\gamma _{{s_k}}^{\left( {{{\cal D}_{k - 1}}} \right)} + \Delta } \right)^2}\\
     \Rightarrow & 2\Delta \gamma _{{s_{k + 1}}}^{\left( {{{\cal D}_{k - 1}}} \right)} \ge 2\Delta \gamma _{{s_k}}^{\left( {{{\cal D}_{k - 1}}} \right)} \Rightarrow \gamma _{{s_{k + 1}}}^{\left( {{{\cal D}_{k - 1}}} \right)} \ge \gamma _{{s_k}}^{\left( {{{\cal D}_{k - 1}}} \right)},
    \end{aligned}
  \end{equation}
  where $\Delta$ is written as
  \begin{equation}
  \begin{aligned}
    \Delta  = & \sum\limits_{f \in {{\cal F}_{{s_{k + 1}}}} \cap {{\cal F}_{{s_k}}}} \\
    ~ & {\left( {\frac{1}{{\left( {d_f^{\left( {{{\cal D}_{k - 1}}} \right)} - 2} \right) + \lambda }} - \frac{1}{{\left( {d_f^{\left( {{{\cal D}_{k - 1}}} \right)} - 1} \right) + \lambda }}} \right)}  > 0.
  \end{aligned}
  \end{equation}
  Thus, the necessary optimality conditions in \emph{Theorem \ref{Theroem_SINR_polar_scheduling_multiple_users}} has been proved.
\end{proof}

We want to point out that above necessary optimality conditions cannot uniquely determine the optimal solution while it will cover the suspect answers which include the optimal solution. In other words, the conditions in Eq. (\ref{SINR_polar_scheduling_multiple_users_formula}) are not sufficient for determining the optimal polar scheduling strategy. However, the following property can be derived from \emph{Theorem \ref{Theroem_SINR_polar_scheduling_multiple_users}}.

\begin{remark}\label{remark_of_theorem_1}
\emph{
  Given a detecting order $\pi$, swapping two consecutive positions $s_k$ and $s_{k+1}$ that satisfy the necessary optimality conditions will lead to a lower variance ${V_\pi }\left( \cal T \right)$ of the detecting SINR sequence $\cal T$.
  }
\end{remark}

The property in \emph{Remark \ref{remark_of_theorem_1}} can significantly reduce the trials to find the optimal PC-NOMA detecting order with respect to the brute-fore search of all possible $\left( V! \right)$ orders. In Fig. \ref{Orders_Bar}, we demonstrate the percentage of the satisfactory orders, which meet the necessary optimality conditions, and the optimal orders with maximum variance ${V_\pi }\left( \cal T \right)$ of the detecting SINR sequence in all possible $\left( V! \right)$ combinations. These two types of orders are denoted by the ``satisfactory orders in all'' and the ``optimal orders in all'', respectively. Additionally, SCMA and PDMA stand for the regular and irregular connections in the factor graph, respectively, which has been introduced in subsection I.A \cite{PDMA_original_TVT}. The $F \times V$ in Fig. \ref{Orders_Bar} represents the configuration of PREs and users in NOMA system. Clearly, the optimal order is not unique in general cases. There exist many equivalent solutions for the optimal order. Furthermore, the larger the system is, the more obvious benefits can be achieved by the necessary optimality conditions. It can exclude most invalid orders so that the search complexity is significantly reduced. We also note that this necessary optimality conditions can work more efficiently in PDMA than that in SCMA.

\begin{figure}[t]
\vspace{0.3em}
\setlength{\abovecaptionskip}{0.cm}
\setlength{\belowcaptionskip}{-0.cm}
  \centering{\includegraphics[scale=0.85]{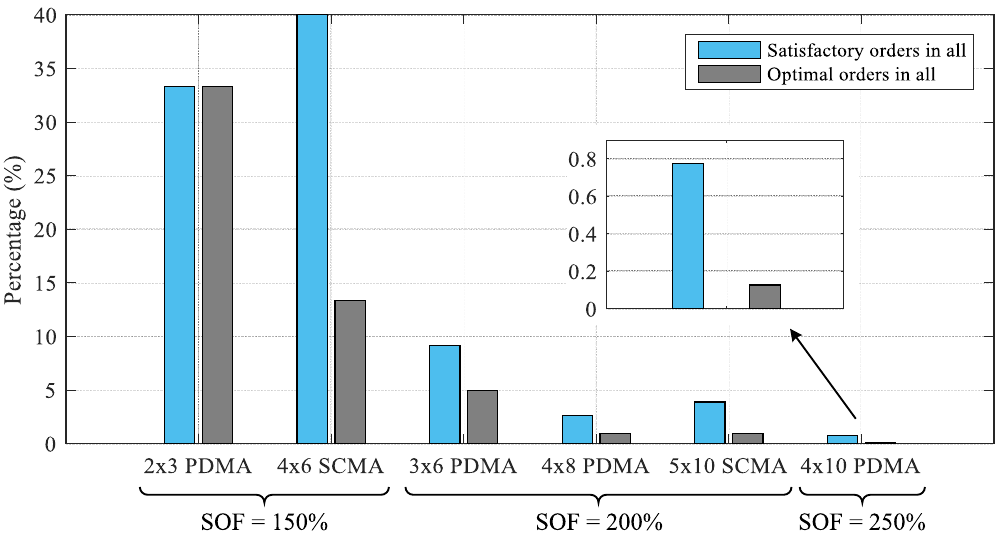}}
  \caption{The percentage of the satisfactory orders, which meet the necessary optimality conditions, in all possible $\left( V! \right)$ combinations, and the optimal orders with maximum variance ${V_\pi }\left( \cal T \right)$ of the detecting SINR sequence in all possible $\left( V! \right)$ combinations.}\label{Orders_Bar}
\vspace{-1em}
\end{figure}

\subsection{Three-Stage Channel Transform under PUP}

In analogy to the approach of SUP, another implementation of the three-stage channel transform is based on the \emph{parallel user partition} (PUP) structure, which corresponds to the first stage in Fig. \ref{SUP_three_stage_channel_transform_figure}. The whole procedure is written as
\begin{equation}\label{PUP_three_stage_channel_transform_procedure}
  W \to \left\{ {{{\overline W}_{{v}}}} \right\} \to \left\{ {{{\overline W}_{{v},j}}} \right\} \to \left\{ {{\overline W}_{{v}}^{\left( i \right)}} \right\}.
\end{equation}

\noindent{\emph{Stage 1: Parallel user partition}}

In the first stage, we parallel split the NOMA channel $W$ into $V$ independent channels ${{\overline W}_{v}} : {{\mathbb{B}}^J} \mapsto {\cal Y}$, whose transition probability functions are written as
\begin{equation}\label{transition_prob_PUP_user_synthesized_channel}
  \begin{aligned}
    & {{\overline W}_v}\left( {{\bf{y}}\left| {{{\bf{b}}_v},{\bf{h}}} \right.} \right) = \\
    & \sum\limits_{\left( {{{\bf{b}}_1}, \cdots ,{{\bf{b}}_{v - 1}},{{\bf{b}}_{v + 1}}, \cdots ,{{\bf{b}}_V}} \right) \in {^{J\left( {V - 1} \right)}}} {\left( {\frac{1}{{{2^{J\left( {V - 1} \right)}}}} \cdot W\left( {{\bf{y}}\left| {{\bf{b}},{\bf{h}}} \right.} \right)} \right)}.
    \end{aligned}
\end{equation}

\begin{definition}
  \emph{An order-$V$ \emph{parallel user partition} ($V$-PUP) of NOMA channel $W$ is defined as
    \begin{equation}\label{PUP_channel_transform}
      W \to \left\{ {{{\overline W}_{{1}}},{{\overline W}_{{2}}}, \cdots ,{{\overline W}_{{V}}}} \right\},
    \end{equation}
  which parallel maps $W$ to a group of $V$ binary vector input channels ${{\overline W}_{{v}}}$, $v=1,2,\cdots,V$, which is referred as the \emph{user synthesized channels}.}
\end{definition}

By assuming that the input binary vectors are equiprobable, the symmetric capacity of ${{\overline W}_{v}}$ is written as
\begin{equation}\label{PUP_synthesized_channel_channel_capacity}
  I\left( {{{\overline W}_v}} \right) = I\left( {{{\sf B}_v};{\sf Y}\left| {\sf H} \right.} \right),
\end{equation}
which indicates these users' data streams are correspondingly detected in a parallel interference cancellation (PIC) manner at the receiver. In addition, suppose the user partition order is configured as $\pi  = \left\{ {1,2, \cdots ,V} \right\}$ in SUP, we obviously have the following inequality
\begin{equation}
  \begin{aligned}
    I\left( {{{\overline W }_v}} \right) & = I\left( {{{\sf B}_v};{\sf Y}\left| {\sf H} \right.} \right)\\
     ~ & \le I\left( {{{\sf B}_v};{\sf Y}\left| {{{\sf B}_1},{{\sf B}_2}, \cdots ,{{\sf B}_{v - 1}},\sf H} \right.} \right) = I\left( {{W_v}} \right).
    \end{aligned}
\end{equation}
Therefore, it can be derived that
\begin{equation}\label{PUP_capacity_small}
  \sum\limits_{v = 1}^V {I\left( {{{\overline W }_v}} \right)}  \le \sum\limits_{v = 1}^V {I\left( {{W_v}} \right)} = I\left( W \right).
\end{equation}
For any other partition order $\pi$, the sum of user synthesized channels' capacities under SUP always equals the original NOMA channel's capacity. Hence, the inequality in Eq. (\ref{PUP_capacity_small}) holds on for any $\pi$.

To describe the distribution of $\{ {{\overline W}_v}\} $, we are also interested in the capacities' mean ${\overline M}\left( W \right)$ and variance ${\overline V}\left( W \right)$, which are defined in analogy to Eq. (\ref{SUP_mean_definition}) and Eq. (\ref{SUP_variance_definition}). However, since all the user synthesized channels are parallel obtained in the PUP approach, the mean and variance of $\{ {{\overline W}_v}\} $ only rely on the original NOMA channel $W$ rather than the partition order in SUP. Furthermore, we note that the mean ${\overline M}\left( W \right)$ is smaller than the corresponding $M_{\pi}\left( W \right)$ in SUP, i.e.,
\begin{equation}\label{PUP_mean_small}
  \overline M \left( W \right) \le {M_\pi }\left( W \right),
\end{equation}
which follows Eq. (\ref{PUP_capacity_small}). Therefore, the PIC multiuser detection in the PUP based PC-NOMA scheme reduces the processing latency while it brings in some capacity loss in the NOMA channel transform and thus introduces performance loss at the same time.

\noindent{\emph{Stage 2: Signal partition}}

After performing the first-stage channel transform which is introduced by NOMA transmission in Eq. (\ref{transition_prob_PUP_user_synthesized_channel}), ${{\overline W}_{v}}$ will be further transformed into a set of independent binary input channels $\{ {{\overline W}_{{v},j}}\} $, which is named as the \emph{signal partition}. In analogy to the three-stage channel transform under SUP, here we also introduce the BICM scheme in signal partition. So the user synthesized channel ${{\overline W}_{v}}$ can be transformed into a group of \emph{bit synthesized channels} $\{ {{\overline W}_{{v},j}}\} $, whose transition probabilities are characterized by
\begin{equation}\label{transition_prob_PUP_bit_synthesized_channel}
  \begin{aligned}
     {{\overline W}_{v,j}}\left( {{\bf{y}} \left| {{b_{v,j}} = b,{\bf{h}}} \right.} \right) = \sum\limits_{\begin{smallmatrix}
     {{\bf{b}}_{v}} \in {{\mathbb B}^J},\\
     {b_{v,j}} = b
     \end{smallmatrix}} {\left( {\frac{1}{{{2^{J - 1}}}} \cdot {{\overline W}_{v}}\left( {{\bf{y}},\left| {{{\bf{b}}_{{v}}},{\bf{h}}} \right.} \right)} \right)},
    \end{aligned}
\end{equation}
where $b_{v,j}$ denotes the $j$-th bit of ${\bf b}_{v}$ with $j = 1,2,\cdots,J$, and $b \in {\mathbb B}$.

\noindent{\emph{Stage 3: Bit partition}}

In the third stage, the channel transform is named as the \emph{bit partition}, and the resulting BMCs $\{ {\overline W}_{v}^{\left( i \right)}\} $ are referred as the \emph{bit polarized channels}. Similar to the bit polarization in the second stage described in \cite{Polar_coded_modulation_seidl}, for each user, by performing the binary-input channel polarization $\bf{G}$ on $N/J$ uses of the resulting $J$ parallel BMCs $\{ {\overline W} _{{v},j}\}$, a total of $NV$ polarized BMCs $\{ {\overline W}_{v}^{\left( i \right)}\} $ will be obtained, where $v = 1,2, \cdots ,V$ and $i = 1,2, \cdots ,N$.

For the three-stage channel transform under PUP, we also expect a lager variance over the user synthesized channel capacities following the polarization principle. Since the PUP scheme is fixed, the only breakthrough is to optimize the NOMA factor graph's structure. But a general comparison on the variance ${\overline V}\left( W \right)$ of different NOMA channels may not be fair because their mean values ${\overline M}\left( W \right)$ are different. Although the variance cannot be viewed as the unique evaluation criterion, it can also reflect the polarization effect among $\{ {{\overline W}_v}\} $ to some extent. The subsequent numerical and simulation results indicate that the irregular connections in PDMA will contribute more to the polarization effect among the user synthesized channels with respect to the regular connections in SCMA.

\section{The Proposed PC-NOMA Schemes}

In this section, two types of PC-NOMA implementations are proposed, which are based on SUP and PUP, respectively. For the SUP based PC-NOMA, we give the detailed construction process and the JSC detecting and decoding scheme, which is based on the ``worst-goes-first'' scheduling algorithm. For the PUP based PC-NOMA scheme, we present the construction process. Then, the PSC detecting and decoding approach is correspondingly developed.

\begin{figure*}[t]
\setlength{\abovecaptionskip}{0.cm}
\setlength{\belowcaptionskip}{-0.cm}
  \centering{\includegraphics[scale=0.58]{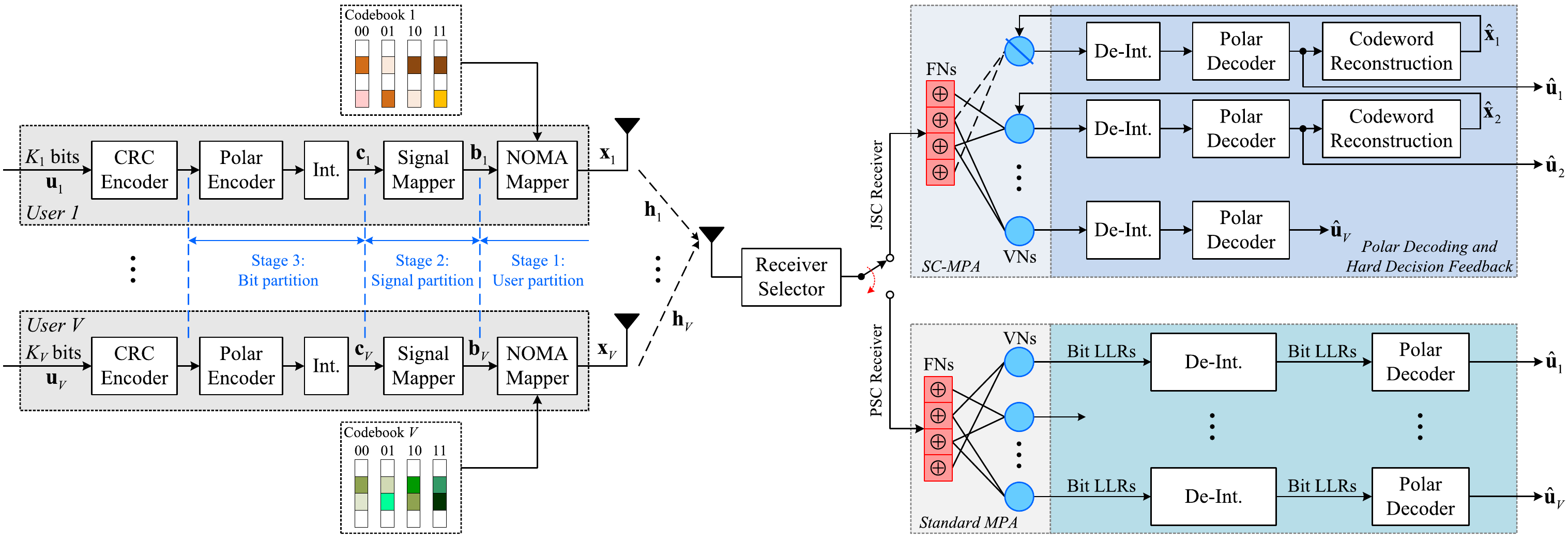}}
  \caption{Illustration of an uplink code-domain multiplexing PC-NOMA system. In the transmitting end, the NOMA mapper, signal mapper and polar encoder correspond to the three stage channel transforms, respectively. There are two types of receivers, which correspond to the SUP and PUP based PC-NOMA schemes, respectively. In this figure, ``Int.'' and ``De-Int.'' denote the abbreviations for interleaver and de-interleaver, respectively.}\label{PCNOMA_system_model}
\end{figure*}

\subsection{The SUP Based PC-NOMA Scheme}

The proposed PC-NOMA scheme is illustrated in Fig. \ref{PCNOMA_system_model}. At the transmitter, the NOMA mapper, the signal mapper and the polar encoder correspond to the three stage channel transforms, respectively. At the receiver, the JSC detecting and decoding scheme is developed, in which the SC-MPA detecting, polar decoding and hard decision feedback are included.

\noindent{\emph{\textbf{1) Transmitter Design: Code Construction}}

Recall that the three-stage channel transform is applied in the SUP based PC-NOMA, hence the construction process can be described by two steps as following, where the Step 1 covers the first and second stage channel transforms.

\noindent{\emph{Step 1) Reliability calculation of bit synthesized channels:}}

This step calculates the reliabilities of the bit synthesized channels which are obtained after the first two-stage channel transforms. Given the received signal $\bf{y}$ and each user's channel gain $\{ {\bf h}_v \}$, the NOMA detector delivers the soft messages of each user's coded bits, i.e., log-likelihood ratios (LLRs). Based on the extrinsic information transfer (EXIT) idea \cite{EXIT}, let $B_{v,j}$ and $A_{v,j}$ denote the random variables which correspond to the $v$-th user's $j$-th coded bit and its \emph{a priori} information. Then for any $v$ and $j$, $A_{v,j}$ is modeled as a conditional Gaussian random variable for $b_{v,j} \in \mathbb{B}$
\begin{equation}
  {A_{v,j}} = {\mu _{v,j}}\left( {1 - 2{b_{v,j}}} \right) + {n_{v,j}},
\end{equation}
where $n_{v,j}$ is a Gaussian random variable with zero mean and variance $\sigma _{v,j}^2 = 2{\mu _{v,j}}$. Using this model, the AMI between ${B_{v,j}}$ and ${A_{v,j}}$ can be written as
\begin{equation}\label{IA}
  \begin{aligned}
{I_A}\left( {v,j} \right) & = I\left( {{B_{v,j}};{A_{v,j}}} \right) = \Omega \left( {{\sigma _{v,j}}} \right)\\
~ & = 1 - \int_{\mathbb{R}} {\frac{1}{{\sqrt {2\pi } {\sigma _{v,j}}}}{{e}^{ - \frac{{{{\left( {l - \sigma _{v,j}^2/2} \right)}^2}}}{{2\sigma _{v,j}^2}}}}{{\log }}\left( {1 - {{e}^{ - l}}} \right){\rm d}l}.
\end{aligned}
\end{equation}
Let $E_{v,j}$ denote the extrinsic information of the $v$-th user's $j$-th coded bit. Once $A_{v,j}$ and received signal $\bf{y}$ have been obtained, the Monte-Carlo simulation is used to estimated the probability distribution function (PDF) of $p\left( {{E_{v,j}}\left| {{b_{v,j}}} \right.} \right)$. So the AMI between ${B_{v,j}}$ and ${E_{v,j}}$ is calculated as
\begin{equation}\label{IE}
  \begin{aligned}
& {I_E}\left( v,j \right) = I\left( {{B_{v,j}};{E_{v,j}}} \right) \\
& = \sum\limits_{{b_{v,j}} \in \mathbb{B}} {\int_{\mathbb{R}} {p\left( {l\left| {{b_{v,j}}} \right.} \right)p\left( {{b_{v,j}}} \right)} {{\log }}\frac{{p\left( {l\left| {{b_{v,j}}} \right.} \right)}}{{\sum\limits_{{{b}_{v,j}'} \in \mathbb{B}} {p\left( {l\left| {{{b}_{v,j}'}} \right.} \right)p\left( {{{b}_{v,j}'}} \right)} }}} {\rm d}l.
\end{aligned}
\end{equation}

For the proposed PC-NOMA scheme under SUP, these coded bits are decided user by user with the SIC manner in the NOMA detector according to Eq. (\ref{transition_prob_SUP_user_synthesized_channel}), hence, the \emph{a priori} information ${I_A}\left( v,j \right)$ can only be $0$ or $1$. If ${I_A}\left( v,j \right) = 0$, it indicates that the $v$-th user has not been decoded. On the contrary, ${I_A}\left( v,j \right) = 1$ stands for that the $v$-th user has already been decoded, where $j = 1,2,\cdots,J$.

Given the partition order $\pi  = \{ {s_1}, \cdots, {s_k}, \cdots ,{s_V}\} $, in order to obtain the $s_k$-th user's reliabilities of bit synthesized channels $\{W_{s_{k},j}\}$, we estimate ${I_E}\left( s_{k},j \right) $ on condition that ${I_A}\left( s_{k'},j \right) = 1$ with $k' = 1, \cdots ,k - 1$, and ${I_A}\left( s_{k'},j \right) = 0$ with $k' = k, \cdots ,V$, where $j = 1,2,\cdots,J$. After that, $W_{{s_k},j}$ will be approximated by the binary input AWGN (BI-AWGN) channel $\widetilde W_{{s_k},j}$ with an equivalent capacity
\begin{equation}\label{equip_cap}
  I( {{W_{{s_k},j}}} ) = I( {{{\widetilde W}_{s_k,j}}} ).
\end{equation}
Therefore, the LLR mean of $\widetilde W_{s_k,j}$ can be calculated as
\begin{equation}\label{LLR_mean}
  {{\tilde \mu }_{s_k,j}} = \frac{{\tilde \sigma}_{s_k,j}^2}{2} = \frac{\left( {{\Omega ^{ - 1}}\left( {{I_E}\left( s_k,j \right)} \right)} \right)^2}{2},
\end{equation}
where ${\Omega ^{ - 1}}\left(  \cdot  \right)$ denotes the inverse function of ${\Omega}\left(  \cdot  \right)$.

\noindent{\emph{Step 2) Reliability calculation of bit polarized channels:}}

Since $W_{s_k,j}$ is approximated by a BI-AWGN channel $\widetilde W_{s_k,j}$ with equivalent capacity, the reliabilities of the corresponding bit polarized channels $\{ W_{s_k}^{\left( i \right)}\} $ can be evaluated by Gaussian approximation (GA) with the LLR mean ${{\tilde \mu }_{s_k,j}}$ in the same way as that in Ar{\i}kan's conventional binary-input channel polarization scheme \cite{Trifonov_construct,Chen_Kai_parallel_construction}. Finally, the $K$ most reliable channels among $\{ W_{s_k}^{\left( i \right)}\} $ are selected to carry the information bits, and the others are fixed to the frozen bits. In the SUP based PC-NOMA scheme, each user will be allocated different number of information bits. For the $s_k$-th user, its information set is denoted by ${\cal A}_{s_k}$. The $s_k$-th user's code rate is ${R_{s_k}} = \left| {{\cal A}_{s_k}} \right|/N$, and the total transmission bits are written as $K = \sum\nolimits_{k = 1}^{V} {\left| {{{\cal A}_{{s_k}}}} \right|}  = \sum\nolimits_{k = 1}^{V} {{K_{{s_k}}}}$.

After above two steps of PC-NOMA construction, the allocation of code rate for the $s_k$-th user matches its user synthesized channel capacity ${I( {{W_{{s_k}}}})}$. For different NOMA channel $W$ and partition order $\pi$, these capacities may vary significantly. The SUP based PC-NOMA scheme preferably matches this scenario by a pretty flexible choice of each user data stream's code rate.

\noindent{\emph{\textbf{2) Partition Order Scheduling}}}

Another problem of the transceiver design is scheduling each user's partition order so as to determine their reliabilities and guide the multiuser detection in the receiver. Since the optimal scheduling strategy will be computationally intensive for $V > 2$, we therefore introduce a search algorithm as the near-optimal solution for this problem, which follows the necessary optimality conditions proved in section III.B and is referred as the polar scheduling (PS) strategy. During the PS process, each bit synthesized channel's LLR mean $\{ {{\tilde \mu }_{{s_k},j}}\}$ can be also determined, which will further be used to estimate the bit polarized channels' reliabilities under the GA method and accomplish the code construction. The PS procedure and LLR means $\{ {{\tilde \mu }_{{s_k},j}}\}$ calculation of the SUP based PC-NOMA scheme are described together in Algorithm \ref{Polar_scheduling_algorithm}.

The search in Algorithm \ref{Polar_scheduling_algorithm} essentially follows the ``worst-goes-first'' principle, which is in contrast with the ``best-goes-first'' based average scheduling (AS) in traditional multiuser detections \cite{D_Tse_wireless_communications}. AS strategy makes the best user be detected first, and the worst user is detected lastly. By using Algorithm \ref{Polar_scheduling_algorithm}, the computational complexity has been significantly reduced compared to the optimal brute-fore search for all possible combinations with $\left( V! \right)$ enumeration complexity, while the PS algorithm needs only $\frac{{\left( {1 + V} \right)V}}{2} - 1$ enumerations. Although the search algorithm in PS is just near-optimal, the subsequent numerical results indicates that it is almost optimal.

\begin{algorithm}[h]
\setlength{\abovecaptionskip}{0.cm}
\setlength{\belowcaptionskip}{-0.cm}
\caption{Polar scheduling (PS) algorithm}\label{Polar_scheduling_algorithm}
\KwIn {NOMA channel $W$ and modulation order $J$;}
\KwOut {NOMA detecting order $\pi$ and the LLR means $\{ {{\tilde \mu }_{{s_k},j}}\} $ of the bit synthesized channels;}
\For{$v = 1 \to V$}
{
    \For{$j = 1 \to J$}
    {
        Initialize ${I_A}\left( {v,j} \right) = 0$\;
    }
}
Initialize $\pi = \varnothing$\;
\For{$k = 1 \to V$}
{
    Calculate each user's ${I_E}\left( v,j \right)$ via the Monte-Carlo simulation, and record ${\tilde I}({W}_v)$ as
    \begin{equation}
      \tilde I\left( {{W_v}} \right) = \sum\nolimits_{j = 1}^J {{I_E}\left( {v,j} \right)},
    \end{equation}
    where $v = 1,2,\cdots,V$\;
    Record the $k$-th detected user index
    \begin{equation}\label{polar_scheduling_key_formula}
      s_k = \mathop {\arg \min }\limits_{v,v \notin \pi } {\tilde I}({W}_v),
    \end{equation}
    and add $s_k$ into $\pi$ such that $\pi  = \{ {s_1}, \cdots ,{s_{k - 1}},{s_k}\} $\;
    \For{$j = 1 \to J$}
    {
        Update ${I_A}\left( {s_k,j} \right) = 1$\;
        Record $I({W_{{s_k},j}}) = {I_E}\left( {{s_k},j} \right)$\;
        Calculate the LLR mean of ${W_{{s_k},j}}$ as
        \begin{equation}
          {{\tilde \mu }_{{s_k},j}} = \frac{{{{\left( {{\Omega ^{ - 1}}\left( {I\left( {{W_{{s_k},j}}} \right)} \right)} \right)}^2}}}{2},
        \end{equation}
        where ${\Omega ^{ - 1}}\left(  \cdot  \right)$ is the inverse function of ${\Omega}\left(  \cdot  \right)$ defined in Eq. (\ref{IA})\;
    }
}
\end{algorithm}

\noindent{\emph{\textbf{3) Receiver Design: JSC Detecting and Decoding}}}

In the receiver, we use a joint successive cancellation (JSC) detecting and decoding structure to construct the multiuser receiver of PC-NOMA. Given the received signal and NOMA detecting order, after the NOMA decoder gives the bit LLRs of $s_k$-th user, they will be sent into the binary polar decoder to retrieve the information bits ${{\bf{\hat u}}_{s_k}}$. Then ${{\bf{\hat u}}_{s_k}}$ is employed to reconstruct codeword ${{\bf{\hat x}}_{s_k}}$. For obtaining the bit LLRs of next $s_{k+1}$-th user, one needs to apply the interference cancellation in the factor graph with hard decision ${{\bf{\hat x}}_{s_k}}$. Given the NOMA detecting order $\pi$, the whole procedure of JSC detecting and decoding scheme for the SUP based PC-NOMA system is described in Algorithm \ref{JSC_detecting_and_decoding_algorithm}.

As a part of the JSC detecting and decoding scheme, we give detailed description about the successive cancellation aided MPA (SC-MPA), which acts as the NOMA detecting algorithm for PC-NOMA. When ${{\cal D}_k}$ has been determined, the set of VNs connected to FN $f$ is written as ${\cal V}_f^{\left( {{{\cal D}_k}} \right)}$, where ${\cal V}_f^{\left( {{{\cal D}_0}} \right)} = {\cal V}_f$, in the pruned factor graph. Clearly, for any $v \in {{\cal D}_k^c}$, the set of FNs connected to it will still be ${\cal F}_v$. By employing the MPA in the pruned factor graph, we can get the $s_{k+1}$-th user's bit LLRs which is denoted by $\Lambda ({b_{{s_{k+1}},j}})$. This algorithm is referred as SC-MPA, which is described as follows:

Let ${\cal V}_f^{\left( {{{\cal D}_k}} \right)}\backslash v$ denote the set ${\cal V}_f^{\left( {{{\cal D}_k}} \right)}$ with VN $v$ excluded, and ${{\cal{F}}_v}\backslash f$ denote the set ${{\cal{F}}_v}$ with FN $f$ excluded. In addition, we define and calculate the following notations corresponding to the $l$-th SC-MPA inner-loop iteration:
\begin{itemize}
  \item ${\Gamma}_{v \to f}^{\left( l \right)}\left( {{{\bf{x}}_v}} \right)$: The log-likelihood function of the $v$-th user's codeword ${{{\bf{x}}_v}}$, which is sent from VN $v$ to FN $f$.
  \item ${\Gamma}_{f \to v}^{\left( l \right)}\left( {{{\bf{x}}_v}} \right)$: The log-likelihood function of the $v$-th user's codeword ${{{\bf{x}}_v}}$, which is sent from FN $f$ to VN $v$.
\end{itemize}

Given ${\cal D}_{k-1}$, for obtaining the $s_k$-th user's bit LLRs, the procedure of SC-MPA is describe as follows:

\noindent{\emph{Step 1: Initialization}}
\begin{enumerate}
  \item Set $l = 1$ and the maximum number of iterations ${\omega}_{\max}$.
  \item Update the received signal as
        \begin{equation}\label{MPA_received_signal_update}
          {\bf{\hat y}} = {\bf{y}} - \sum\limits_{v \in {{\cal D}_{k - 1}}} {diag\left( {{{\bf{h}}_v}} \right){{{\bf{\hat x}}}_v}}.
        \end{equation}
  \item Initialize the log-domain conditional probability
        \begin{equation}\label{MPA_conditional probability}
        {\Xi _f}\left( {\bf{x}} \right) =  - \frac{1}{{{N_0}}}{\left\| {{{\hat y}_f} - \sum\nolimits_{v \in {\cal V}_f^{\left( {{{\cal D}_{k-1}}} \right)}} {{h_{v,f}}{x_{v,f}}} } \right\|^2},
        \end{equation}
        where ${\bf x} = \{ {{\bf x}_v}\} $ with $v \in {\cal V}_f^{\left( {{{\cal D}_{k-1}}} \right)}$, and ${{\hat y}_f}$ is the $f$-th element in ${\bf{\hat y}}$.
  \item For $v \in {{\cal D}_{k-1}^c}$ and $f \in {\cal F}_v^{\left( {{{\cal D}_{k-1}}} \right)}$, initialize ${\Gamma}_{v \to f}^{\left( 0 \right)}\left( {{{\bf{x}}_v}} \right) =  - \ln M =  - \ln {2^J}$.
\end{enumerate}

\begin{algorithm}[t]
\setlength{\abovecaptionskip}{0cm}
\setlength{\belowcaptionskip}{-1cm}
\caption{JSC detecting and decoding for PC-NOMA}\label{JSC_detecting_and_decoding_algorithm}
\KwIn {Total $N/J$ time slots received signal $\bf y$, each user's channel gain $\{ {\bf h}_v \}$ and the predefined NOMA detecting order $\pi = \{ {s_1},{s_2}, \cdots ,{s_V}\} $;}
\KwOut {Each user's estimated information bits ${{\bf{\hat u}}_{v}}$;}
\For{$k = 1 \to V$}
{
    Set ${{\cal D}_{k-1}} = \left\{ {{s_1},{s_2}, \cdots ,{s_{k-1}}} \right\}$, where ${{\cal D}_0} = \varnothing $\;
    Run SC-MPA in the pruned factor graph and output the $s_k$-th user's bit LLRs\;
    Send the bit LLRs into the binary polar decoder to retrieve ${{\bf{\hat u}}_{s_k}}$ based on SC, SCL or CA-SCL decoding algorithms\;
    Employ the $s_k$-th user's decoded result ${{\bf{\hat u}}_{s_k}}$ to reconstruct its corresponding NOMA codewords ${{{\bf{\hat x}}}_{{s_k}}}$ and feed it back to the NOMA detector\;
}
\end{algorithm}

\noindent{\emph{Step 2: Iterative message passing}}
\begin{enumerate}
  \item
  FN updating: for $1 \le f \le F$ and each VN $v \in {\cal V}_f^{\left( {{{\cal D}_{k-1}}} \right)}$, calculate
  \begin{equation}\label{standard_MPA_FN_update}
  \begin{aligned}
      {\Gamma}_{f \to v}^{\left( l \right)}\left( {{{\bf{x}}_v}} \right) & = {\mathop {{\max }^*}\limits_{{{\bf{x}}_u}:u \in {{\cal V}_f^{\left( {{{\cal D}_{k-1}}} \right)}}\backslash v}}\\
      & \left\{ {{\Xi} _f}\left( {\bf{x}} \right) + \sum\limits_{u \in {{\cal V}_f^{\left( {{{\cal D}_{k-1}}} \right)}}\backslash v} {{\Gamma}_{u \to f}^{\left( {l - 1} \right)}\left( {{{\bf{x}}_u}} \right)}  \right\},
  \end{aligned}
  \end{equation}
  where the ${\max ^*}\left\{ {{a_1}, \cdots ,{a_n}} \right\} = \ln \left( {{e^{{a_1}}} +  \cdots  + {e^{{a_n}}}} \right) \approx \max \left\{ {{a_1}, \cdots ,{a_n}} \right\}$. Then for each ${\Gamma}_{f \to v}^{\left( l \right)}\left( {{{\bf{x}}_v}} \right)$, the normalization will be performed as
  \begin{equation}
    \Gamma _{f \to v}^{\left( l \right)}\left( {{{\bf{x}}_v}} \right) = \Gamma _{f \to v}^{\left( l \right)}\left( {{{\bf{x}}_v}} \right) - {\mathop {{\max }^*}\limits_{{{\bf{x}}_v} \in {{\cal X}_v}}}\left\{ {\Gamma _{f \to v}^{\left( l \right)}\left( {{{\bf{x}}_v}} \right)} \right\}.
  \end{equation}
  \item VN updating: for $k \le k' \le V$, let $v = s_{k'}$, then for each FN $f \in {\cal F}_v$, calculate
  \begin{equation}\label{standard_MPA_VN_update}
    {\Gamma}_{v \to f}^{\left( l \right)}\left( {{{\bf{x}}_v}} \right) = \sum\limits_{h \in {{\cal F}_v}\backslash f} {{\Gamma}_{h \to v}^{\left( l \right)}\left( {{{\bf{x}}_v}} \right)}.
  \end{equation}
  Meanwhile the normalization will also be performed on each ${\Gamma}_{v \to f}^{\left( l \right)}\left( {{{\bf{x}}_v}} \right)$ as
  \begin{equation}
    \Gamma _{v \to f}^{\left( l \right)}\left( {{{\bf{x}}_v}} \right) = \Gamma _{v \to f}^{\left( l \right)}\left( {{{\bf{x}}_v}} \right) - {\mathop {{\max }^*}\limits_{{{\bf{x}}_v} \in {{\cal X}_v}}}\left\{ {\Gamma _{v \to f}^{\left( l \right)}\left( {{{\bf{x}}_v}} \right)} \right\}.
  \end{equation}
  \item Update $l = l + 1$, if $l > {{\omega}_{\max }}$, go to Step 3. Otherwise, go to the beginning of Step 2.
\end{enumerate}

\noindent{\emph{Step 3: The $s_k$-th user's bit LLR calculation}}

After the iterative message passing in Step 2, the $s_k$-th user's bit LLRs are given by
\begin{equation}\label{standard_MPA_LLR_output}
    {\Lambda}\left( {{b_{{s_k},j}}} \right) = {\mathop {{\max}^*}\limits_{{{\bf{x}}_{s_k}}:{{{\bf{x}}_{s_k}^{\left( j \right)}}} = 0}}\left\{ {{\Gamma}_{s_k}}\left( {{{\bf{x}}_{s_k}}} \right)\right\}  - {\mathop {{\max}^*}\limits_{{{\bf{x}}_{s_k}}:{{{\bf{x}}_{s_k}^{\left( j \right)}}} = 1}}\left\{ {{\Gamma}_{s_k}}\left( {{{\bf{x}}_{s_k}}} \right) \right\},
\end{equation}
where ${{\Gamma}_{s_k}}\left( {{{\bf{x}}_{s_k}}} \right)$ is written as
\begin{equation}\label{standard_MPA_LLR_output_Dv}
  {{\Gamma}_{s_k}}\left( {{{\bf{x}}_v}} \right) = \sum\limits_{h \in {{\cal F}_{s_k}}} {{\Gamma}_{h \to {s_k}}^{\left( {{{\omega}_{\max }}} \right)}\left( {{{\bf{x}}_{s_k}}} \right)}.
\end{equation}

Several decoding algorithms can be used in each user's polar decoder. Note that the performance of successive cancellation (SC) decoding \cite{Arikan} is unsatisfying in the practical cases with finite-length blocks, we employ the cyclic redundancy check (CRC) aided successive cancellation list (CA-SCL) decoding \cite{Tal_Vardy_SCL,niukai_com_magazine}. In the JSC detecting and decoding algorithm, since the CA-SCL decoding algorithm is applied, with the help of the CRC sequence, it can be known whether a user has been correctly decoded. If a user has decoding paths passed the CRC check in the list, the CRC-passing path with the largest likelihood will be taken as the decoding result ${{\bf{\hat u}}_{s_k}}$. On the contrary, if the user has no decoding path passed CRC check, the survival path with the largest likelihood will be taken as the decoding result ${{\bf{\hat u}}_{s_k}}$.

In Fig. \ref{PCNOMA_system_model}, we show a toy example for the JSC receiver, where the first user has already been decoded. Then its decoded information bits ${{\bf{\hat u}}_{1}}$ are employed to reconstruct the codeword ${{\bf{\hat x}}_{1}}$. As the hard information, ${{\bf{\hat x}}_{1}}$ will feedback into the NOMA detector. Correspondingly, its associated VN and edges are removed from the factor graph.

\subsection{The PUP Based PC-NOMA Scheme}

In contrast to the SUP approach for PC-NOMA, all the users are treated equally at the transceiver in the framework of PUP based PC-NOMA. As demonstrated in Fig. \ref{PCNOMA_system_model}, the transmitter is still divided into three stages. The PSC detecting and decoding scheme is used in the receiver, where all the users' data streams are parallel detected and decoded.

\noindent{\emph{\textbf{1) Transmitter Design: Code Construction}}}

The construction process can be described by two steps as follows:

\noindent{\emph{Step 1) Reliability calculation of bit synthesized channels:}}

Since the receiver of the PUP based PC-NOMA scheme performs parallel detecting and decoding, i.e., it neglects the feedback of decoded users. Correspondingly, the reliability calculation of the bit synthesized channels are simpler than that in the SUP scenario.

To obtain the reliabilities of the bit synthesized channels $\{{\overline W}_{v,j}\}$, we estimate ${I_E}\left( v,j \right) $ on condition that all the \emph{a priori} information is set to zero, i.e., ${I_A}\left( v,j \right) = 0$ $\forall v,j$.  After that, ${\overline W}_{v,j}$ is approximated by a BI-AWGN channel $\widetilde W_{v,j}$ with equivalent capacity
\begin{equation}\label{equip_cap}
  I( {{{\overline W}_{v,j}}} ) = I( {{{\widetilde W}_{v,j}}} ).
\end{equation}
Then, the LLR mean of $\widetilde W_{v,j}$ can be characterized as
\begin{equation}\label{LLR_mean}
  {{\tilde \mu }_{v,j}} = \frac{{\tilde \sigma}_{v,j}^2}{2} = \frac{\left( {{\Omega ^{ - 1}}\left( {{I_E}\left( v,j \right)} \right)} \right)^2}{2}.
\end{equation}

\noindent{\emph{Step 2) Reliability calculation of bit polarized channels:}}

The reliabilities of the corresponding bit polarized channels $\{ W_{v}^{\left( i \right)}\} $ can be evaluated by the GA approach with the LLR mean ${{\tilde \mu }_{v,j}}$ in the same way as that in \cite{Trifonov_construct,Chen_Kai_parallel_construction}. Finally, the $K$ most reliable channels among $\{ W_{v}^{\left( i \right)}\} $ are selected to carry the information bits, and the others are fixed to the frozen bits. In the PUP based PC-NOMA scheme, each user may also be allocated different number of information bits, which is determined by the NOMA channel $W$. For the $v$-th user, its information set is denoted by ${\cal A}_{v}$. The $v$-th user's code rate is ${R_{v}} = \left| {{\cal A}_{v}} \right|/N$, and we have $K = \sum\nolimits_{v = 1}^{V} {\left| {{{\cal A}_{{v}}}} \right|}  = \sum\nolimits_{v = 1}^{V} {{K_{{v}}}}$.

\noindent{\emph{\textbf{2) Receiver Design: PSC Detecting and Decoding}}}

In the receiver, the standard MPA \cite{SCMA_original_PIMRC} acts as the multiuser detection to parallel eliminate the multiple access interference (MAI). Then the remaining MAI will be further cancellated by each user's polar decoder with the successive cancellation decoding schemes in the second step. On the whole, the two steps of parallel and successive cancellation (PSC) of MAI compose the receiver structure of the PUP based PC-NOMA scheme. The multiuser detector neglects the relations between the users and computes their bit LLRs independently without feedback of other user's information, which aims for reducing the latency. These bit LLRs are then fed into each user's polar decoder, respectively. Compared to the JSC receiver, the PSC can obviously reduce the processing latency with some cost of performance loss.

\section{Performance Evaluation}

In this section, the numerical and simulation results about the PC-NOMA are provided. Three NOMA configurations are adopted, where the factor graph matrices of the $4 \times 6$ SCMA and the $2 \times 3$ PDMA have already been given in Eq. (\ref{6x4_factor_graph_matrix}) and Eq. (\ref{3x2_factor_graph_matrix}), respectively. Another factor graph matrix of the $3 \times 6$ PDMA scheme \cite{PDMA_original_TVT} is defined as
\begin{equation}\label{6x3_factor_graph_matrix}
{\bf{F}} = \left[ {\begin{array}{*{20}{c}}
1&1&0&1&0&0\\
1&0&1&0&1&0\\
0&1&1&0&0&1
\end{array}} \right].
\end{equation}
In addition, for the $4 \times 6$ SCMA scheme, the codebooks in \cite{huawei_codebook} are used, whose design details are given in \cite{huawei_codebook_details}. Note that this SCMA codebook design aims for maximizing the sum rate based on the rotation of pulse-amplitude modulation (PAM) constellations. For the other two PDMA schemes, each user's codebook is the standard quadrature phase shift keying (QPSK) modulation without complex codebook design \cite{PDMA_original_TVT}. So in these three NOMA configurations, each user's modulation order is $J = 2$.

For PC-NOMA, each user's code length is $N = 1024$, and the overall code rate is configured as $R = 0.5$. Additionally, the CRC-16 in \cite{LTE} is used. When the SUP based PC-NOMA scheme is adopted, the JSC detecting and decoding algorithm will be exploited in the receiver, where the SC-MPA detecting and the CA-SCL decoding is employed, and the list size is set to $32$. The PS strategy is used to determine the partition and detecting order, and the AS approach is also adopted as a comparison. When we use the PUP based PC-NOMA scheme, the standard MPA detecting and the CA-SCL decoding will be correspondingly used. The number of inner-loop iterations in the SC-MPA or standard MPA is set to $3$.

As the comparative scheme, the TC-NOMA approach is also investigated. The turbo encoder and rate-matching algorithm used in 3GPP LTE standard \cite{LTE} are employed. Each user's code length and the overall code rate are the same as that in PC-NOMA. In addition, the Log-MAP algorithm \cite{Lin_shu_channel_coding} is used for turbo decoding, where the number of iterations between the two component codes is set to $8$. At this time, the complexity of turbo decoding is higher than polar decoding \cite{niukai_com_magazine}. For fair comparison, the JSC multiuser receiver in PC-NOMA corresponds to the standard MPA receiver of TC-NOMA with outer-loop iterations (w/-oi) between the detector and the decoder \cite{SCMA_outloop_iterations}. Since there are $V$ times hard information feedback in the JSC receiver of PC-NOMA, the number of outer-loop iterations in TC-NOMA is thus configured as $V$. The low latency PSC scheme corresponds to the standard MPA receiver of TC-NOMA without outer-loop iterations (w/o-oi) \cite{SCMA_original_PIMRC}.

\begin{figure}[t]
\vspace{0.3em}
\setlength{\abovecaptionskip}{0.cm}
\setlength{\belowcaptionskip}{-0.cm}
  \centering{\includegraphics[scale=0.68]{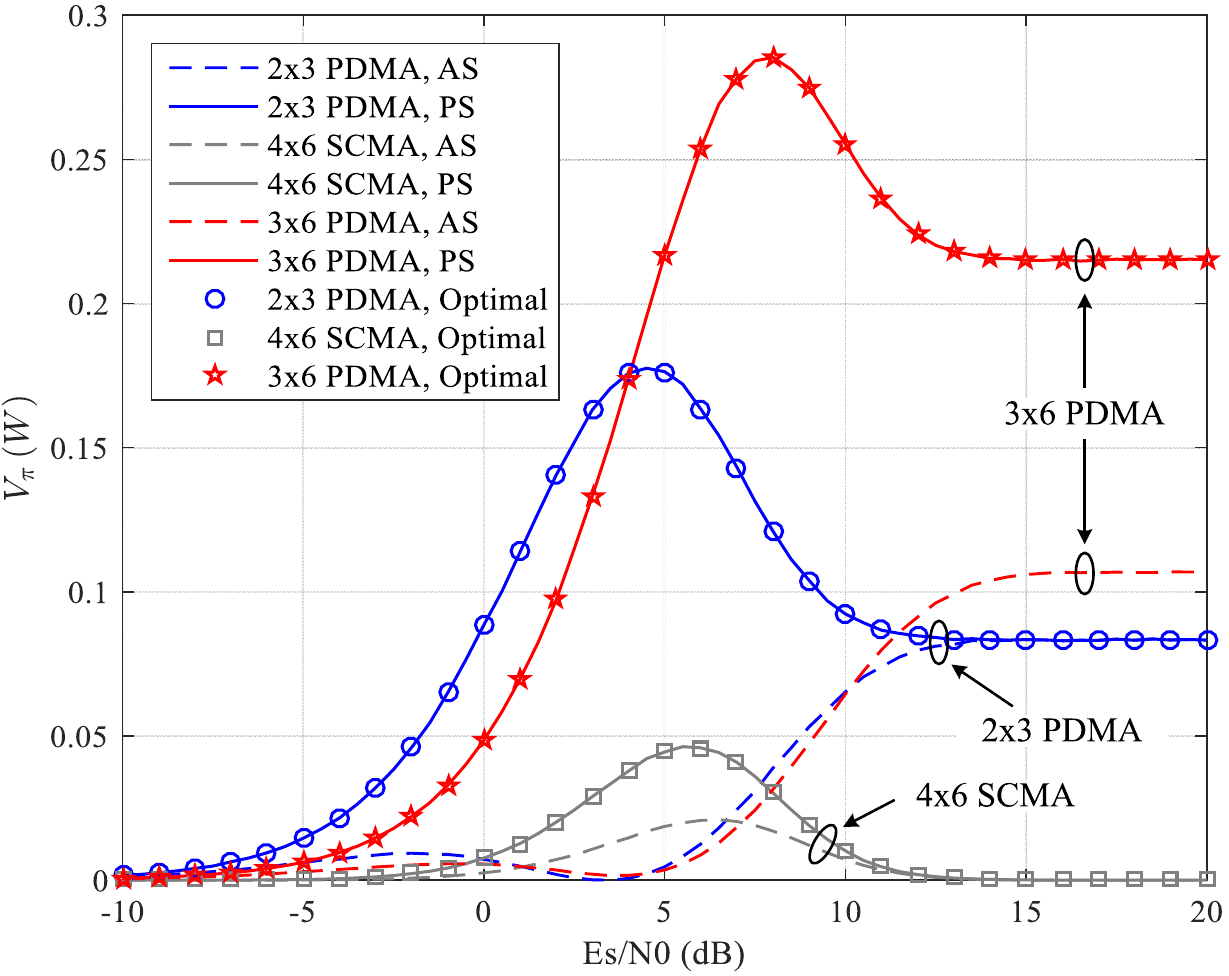}}
  \caption{The variance of user synthesized channel capacities with the SUP scheme. The channel is configured as the AWGN channel.}\label{SIC_AMI_Var_fig}
\end{figure}

\subsection{Numerical Result Analysis}

For the SUP scheme, Fig. \ref{SIC_AMI_Var_fig} depicts the variance of the user synthesized channel capacities ${V_\pi }\left( W \right)$ defined in Eq. (\ref{SUP_variance_definition}) over the AWGN channel. It can be observed that the proposed PS strategy can significantly lead to a larger variance compared to the AS strategy, as expected before. The optimal results are obtained by the brute-force search, which aims for maximizing the variance. It can be seen that the proposed search method in PS is almost the optimum for any configurations of SNRs and NOMA systems. Furthermore, from these results, we can conclude that the irregular superposition structure of PDMA will be more favorable for a larger variance. It makes the PC-PDMA system can further take the advantage of joint design of transmitter and receiver. By contrast, the regular superposition structure in SCMA leads to relatively smaller variance, which indicates the differences among the users are smaller.

\subsection{Simulation Result Analysis}

The average block error ratio (BLER) performance of the SUP based PC-NOMA schemes are shown in Fig. \ref{BLER1_fig}, where the channel is configured as the AWGN and the corresponding TC-NOMA schemes are also evaluated. Clearly, when the SOF $\eta = 150\%$, the proposed PC-SCMA scheme outperforms the TC-SCMA scheme by at least $0.5$dB. Furthermore, when the SOF $\eta = 200\%$, the gap between PC-PDMA and TC-PDMA can even reach about $1.5$dB, which verifies the conclusion in Fig. \ref{SIC_AMI_Var_fig}. In addition, we note that the TC-NOMA schemes demonstrate obvious error floor in the high SNR region, while all of the PC-NOMA schemes can strictly ensure no error floor which follows the principle of polar codes. Besides, compared to the AS manner, the proposed PS strategy can bring additional $0.25$dB gain, which verifies the efficiency of the proposed scheduling algorithm. Additionally, we observe that when the SOF $\eta = 150\%$, the PC-PDMA scheme with standard QPSK modulation has just about $0.1$dB performance loss with respect to the PC-SCMA. But the SCMA codebooks need complex design to achieve this performance. That means PDMA can bring additional polarization effect via the irregular degree distribution and improves its performance. The above analysis indicates that for PC-NOMA, the codebooks and the NOMA superposition structure need to be jointly designed so as to achieve better performance.

For the Rayleigh fading channel, the average BLER results of the SUP based PC-NOMA schemes are demonstrated in Fig. \ref{BLER2_fig}, where the corresponding TC-NOMA schemes are also evaluated. In analogy to Fig. \ref{BLER1_fig}, it can also be observed that the PC-NOMA schemes outperform the TC-NOMA schemes by $0.5 \sim 0.9$dB at various configurations. The TC-NOMA schemes will also show significant error floors in the high SNR regions, while the proposed PC-NOMA schemes can effectively avoid this weakness. In addition, the proposed PS strategy is also efficient in fading channel. Furthermore, it is worth noticing that in fading channel, the PC-PDMA scheme can outperform the PC-SCMA scheme.

Fig. \ref{BLER3_fig} shows the average BLER performance of the PUP based PC-NOMA schemes over the AWGN and the Rayleigh fading channels, where the performances of the corresponding TC-NOMA w/o-oi schemes are also provided. Clearly, under the AWGN channel, given the SOF $\eta = 150\%$, the proposed PC-SCMA scheme outperforms the TC-SCMA about $0.4$dB, where the $4 \times 6$ SCMA is adopted. With the same SOF, when the $2 \times 3$ PDMA is used, this performance gain will expand to about $0.7$dB. This observation further indicates that the proposed PC-PDMA schemes can better fit and exploit the irregular structure of PDMA. It will conversely guide the joint design of NOMA codebook and the channel coding under the framework of PC-NOMA. Under the Rayleigh fading channel, the SOF is configured as $\eta = 150\%$ and $\eta = 200\%$, the performance gain of PC-PDMA achieves about $0.7 \sim 0.8$dB compared to TC-PDMA.

\begin{figure}[htbp]
\vspace{0.3em}
\setlength{\abovecaptionskip}{0.cm}
\setlength{\belowcaptionskip}{-0.cm}
  \centering{\includegraphics[scale=0.58]{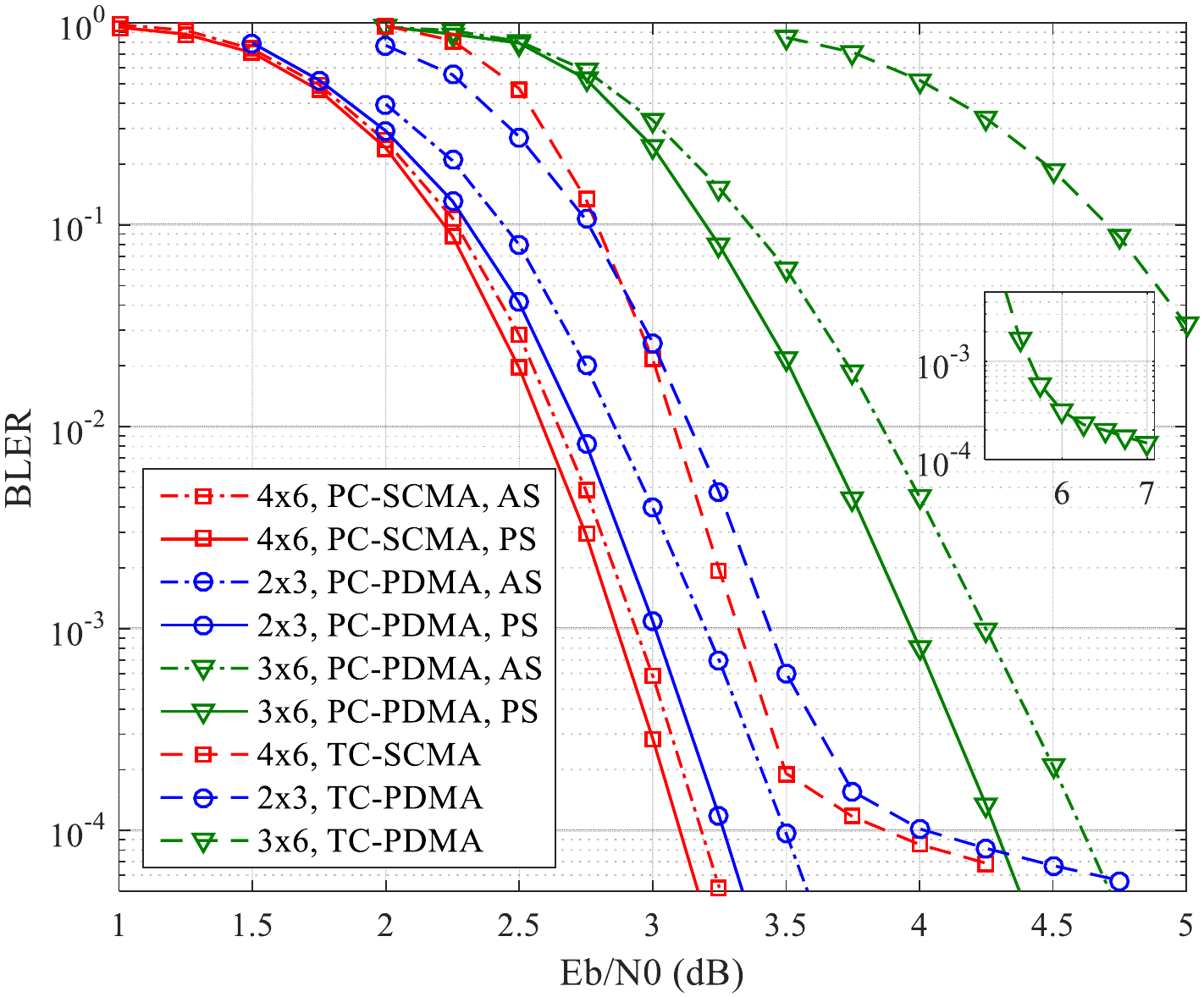}}
  \caption{The average BLER performance under the AWGN channel, where the SUP based PC-NOMA and the TC-NOMA w/-oi are adopted.}\label{BLER1_fig}
\end{figure}

\begin{figure}[htbp]
\setlength{\abovecaptionskip}{0.cm}
\setlength{\belowcaptionskip}{-0.cm}
  \centering{\includegraphics[scale=0.58]{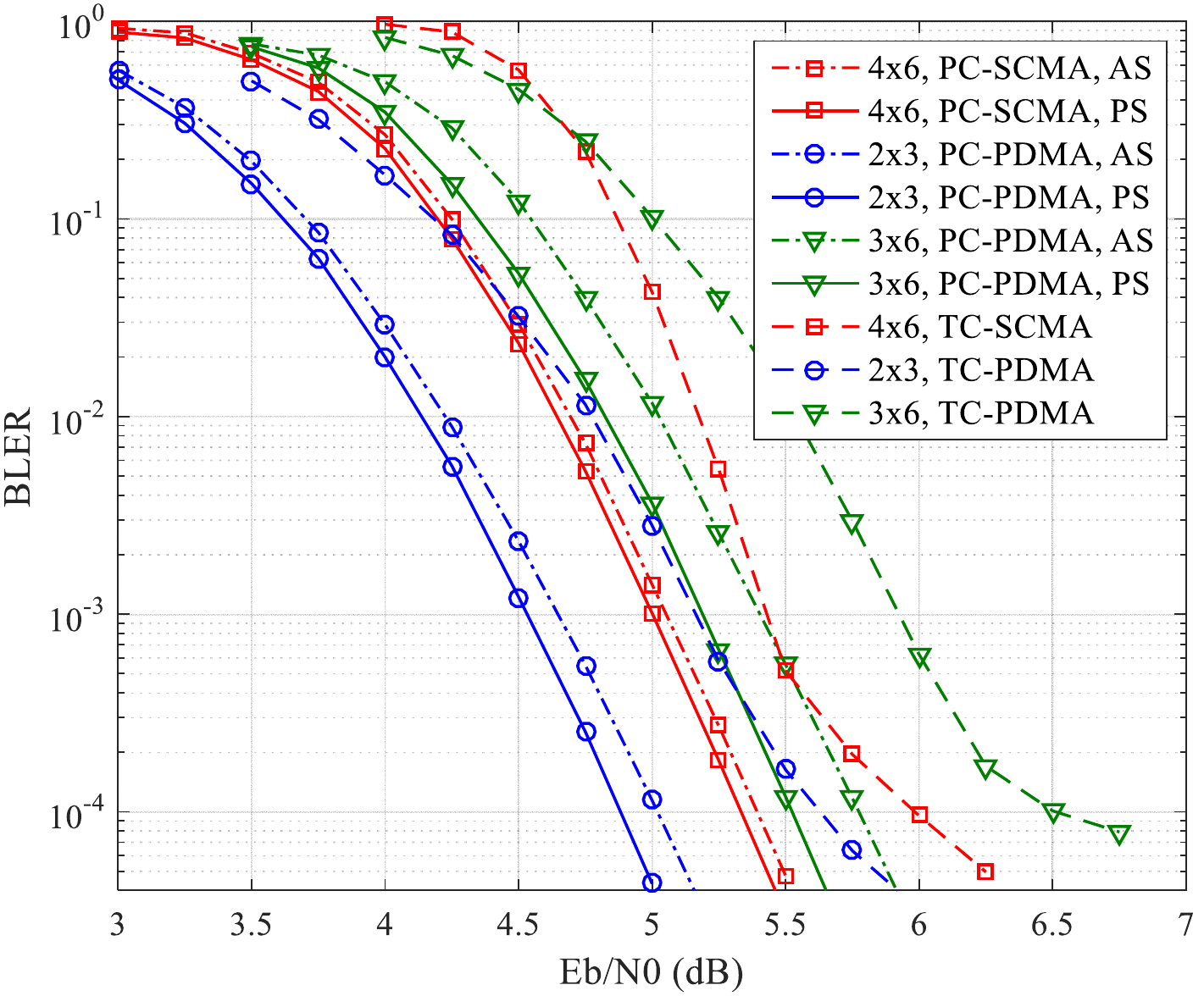}}
  \caption{The average BLER performance under the Rayleigh fading channel, where the SUP based PC-NOMA and the TC-NOMA w/-oi are adopted.}\label{BLER2_fig}
\end{figure}

\begin{figure}[htbp]
\setlength{\abovecaptionskip}{0.cm}
\setlength{\belowcaptionskip}{-0.cm}
  \centering{\includegraphics[scale=0.58]{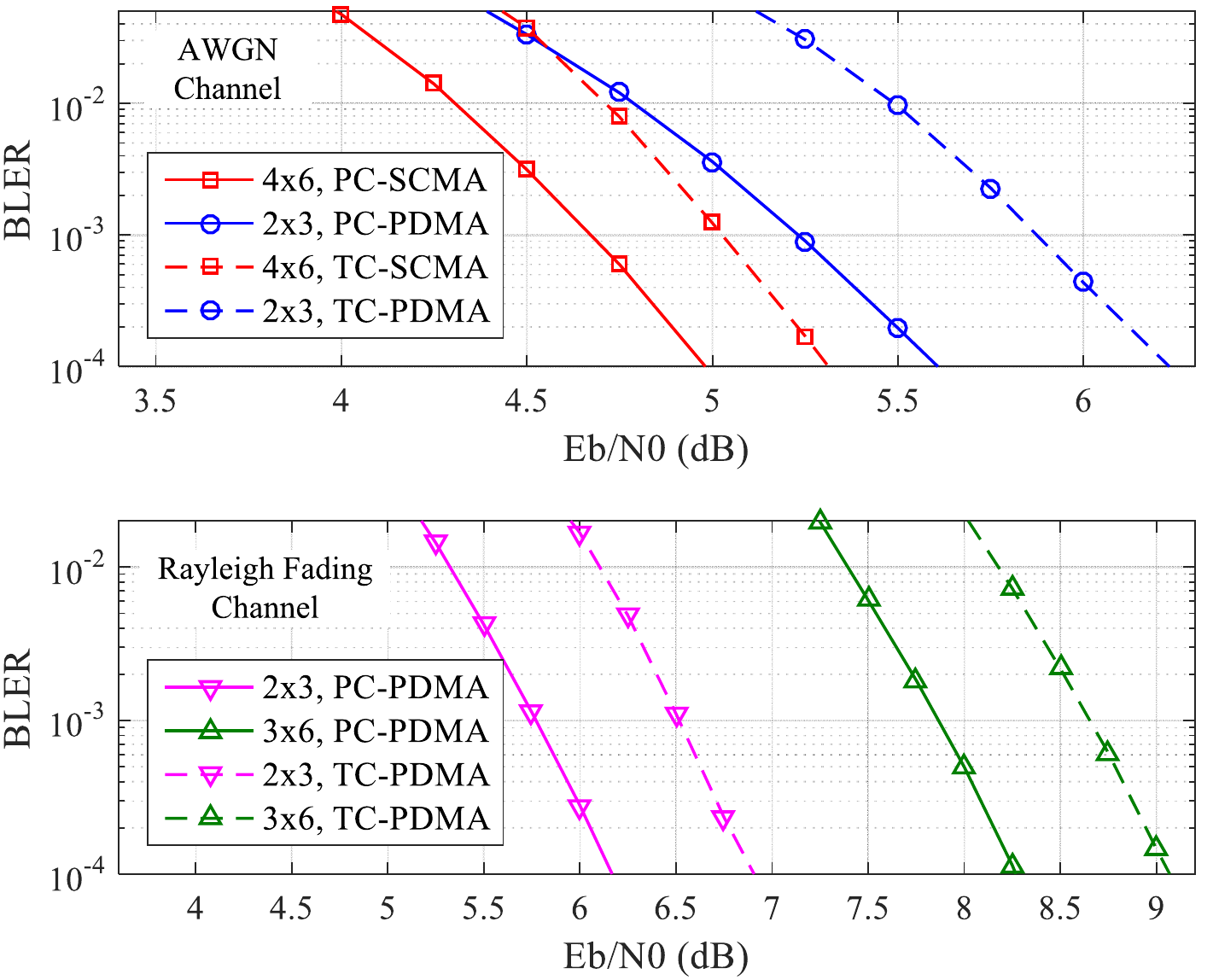}}
  \caption{The average BLER performance comparison under the AWGN and Rayleigh fading channels, where the PUP based PC-NOMA and the TC-NOMA w/o-oi are adopted.}\label{BLER3_fig}
\end{figure}

\section{Conclusion}

In this paper, we establish a framework for a joint design and optimization of the binary polar coding and the code-domain NOMA. To the best knowledge of the authors, this paper is the first one to investigate the PC-NOMA system. We extend the channel polarization idea to NOMA scenario, whereby the NOMA channel is decomposed into a series of binary-input channels under the three-stage channel transform. From the perspective of information theory, two practical PC-NOMA schemes are proposed. For the SUP based PC-NOMA scheme, we develop a new polar scheduling strategy to determine the NOMA detecting order so as to enhance the polarization effect among the user synthesized channels, which helps to improve the system performance. Then, the JSC detecting and decoding scheme is proposed to construct the multiuser receiver of the SUP based PC-NOMA. For the PUP based PC-NOMA scheme, the PSC detecting and decoding scheme is used to reduce the processing latency. The numerical and simulation results show the PC-NOMA schemes can better fit and exploit the irregular superposition structure of NOMA, which indeed indicates that the design of NOMA codebook (or superposition structure) and the channel coding scheme should be taken by a joint manner. This is our future research direction.

\end{document}